\documentclass[journal,twocolumn]{IEEEtran}
\ifCLASSINFOpdf
\else
\fi
\usepackage{amsmath,amsthm,amssymb,bm,soul,cancel,graphicx,booktabs,float,hyperref,multirow,xcolor,subcaption,enumitem}
\usepackage[scr=boondoxo,scrscaled=1.05]{mathalfa}
\usepackage{mdframed}
\usepackage[linesnumbered,ruled,vlined]{algorithm2e}
\usepackage[framed,numbered,autolinebreaks,useliterate]{mcode}

\newtheorem{theorem}{Theorem}
\newtheorem{definition}{Definition}
\newtheorem{corollary}{Corollary}
\newtheorem{condition}{Condition}
\newtheorem{example}{Example}

\begin{document}
%
\title{Kernel Center Adaptation in the Reproducing Kernel Hilbert Space  Embedding Method}
%
%
%

\author{Sai~Tej~Paruchuri, 
        Jia~Guo, 
        and~Andrew~Kurdila 
\thanks{The authors are with the Mechanical Engineering Department, Virginia Tech, Blacksburg,
VA, 24061 USA e-mail: saitejp@vt.edu, jguo18@vt.edu, kurdila@vt.edu.}
}

%
%

\markboth{Journal of \LaTeX\ Class Files,~Vol.~14, No.~8, August~2015}%
{Shell \MakeLowercase{\textit{et al.}}: Bare Demo of IEEEtran.cls for IEEE Journals}
%



\maketitle

\begin{abstract}
The performance of adaptive estimators that employ embedding in reproducing kernel Hilbert spaces (RKHS) depends on the choice of the location of basis kernel centers. Parameter convergence and error approximation rates depend on where and how the kernel centers are distributed in the state-space. In this paper, we develop the theory that relates parameter convergence and approximation rates to the position of kernel centers. We develop criteria for choosing kernel centers in a specific class of systems - ones in which the state trajectory regularly visits the neighborhood of the positive limit set. Two algorithms, based on centroidal Voronoi tessellations and Kohonen self-organizing maps, are derived to choose kernel centers in the RKHS embedding method. Finally, we implement these methods on two practical examples and test their effectiveness.  
\end{abstract}


%
\IEEEpeerreviewmaketitle

\section{Introduction}
%
%
%
%
\IEEEPARstart{A}{daptive} estimation of unknown nonlinearities appearing in dynamical systems is a topic that has been studied over the past four decades. The finite-dimensional versions of such problems are described in classical texts like \cite{Ioannou, Sastry2011, Narendra2012}. The goal of these methods is to estimate an unknown term appearing in the governing ordinary differential equations (ODEs). A common assumption in such problems is that all the states are available for measurement. Many of these methods also assume that the unknown function belongs to some hypothesis space of functions. The particular class of adaptive estimators studied in this paper assumes that the hypothesis space is a reproducing kernel Hilbert space (RKHS). An RKHS $\mathcal{H}_{\mathbb{R}^d}$ is a Hilbert space of functions on the state-space $\mathbb{R}^d$ that is defined in terms of a positive-definite kernel $\mathcal{K}: \mathbb{R}^d \times \mathbb{R}^d \to \mathbb{R}$. An example of an RKHS is the space generated by the Gaussian radial basis kernels that have the form $\mathcal{K}(x,y) := e^{ \zeta \| x - y\|^2}$, where $\zeta$ is positive. The additional structure induced by the kernel $\mathcal{K}$ on $\mathcal{H}_{\mathbb{R}^d}$ enables the proof of crucial convergence results, even for the infinite-dimensional cases. The finite-dimensional version of the RKHS adaptive estimators have been studied in \cite{Kurdila1995, Kingravi2012a}. However, the results for the infinite-dimensional adaptive estimation cases are relatively new and were investigated by Bobade et al. in \cite{Bobade2019}. 

In both the finite and infinite-dimensional cases, the unknown function $f \in \mathcal{H}$ has the form $f(\cdot) = \sum_i \alpha_i \mathfrak{K}_{\bm{x}_i} (\cdot)$, where $\mathfrak{K}_{\bm{x}_i} (\cdot) := \mathcal{K}(\bm{x}_i, \cdot)$ with $\bm{x}_i \in \mathbb{R}^d$. Note, the index $i \in \{1,\ldots,n\}$ for the $n$-dimensional case while $i \in \mathbb{N}$ for the infinite-dimensional case. We refer to $\mathfrak{K}_{\bm{x}_i} \in \mathcal{H}$ as the kernel function centered at $\bm{x}_i$ or the regressor function. Thus, we express the unknown function $f$ as a linear combination of kernels centered at different points in the state-space. When the set of centers are fixed or held constant, the analysis in \cite{Ioannou, Sastry2011, Narendra2012} are applicable. This paper specifically studies how such centers can be chosen adaptively in the RKHS embedding method. 

The general problem of center selection is familiar in both adaptive estimation and in machine learning methods based on radial basis functions (RBF) networks. Roughly speaking, the primary difference between the problem of center selection in these two applications is that computations are usually static or offline in machine learning, whereas they are recursive or online in adaptive estimation. One of the most common unsupervised learning methods for choosing the kernel centers in RBF networks is the k-mean clustering or Lloyd's algorithm \cite{Mostafa2012, Wu2012a}. Researchers in the machine learning community have developed sophisticated methods for center selection/adaptation to optimize RBF networks. Some of the early accounts of such methods can be found in \cite{Orr1995a, Warwick1995a, Nie1993}. Self-organizing maps are another alternative for clustering data and thereby determining the kernel centers. The technique in \cite{Hager2004} relies on adding kernels such that the sum of squared error is minimized. Lin and Chen describe a method that combines Kohonen self-organizing maps and RBF networks in \cite{Lin2005}. Kernel centers are chosen based on the condition number of the sensitivity matrix in \cite{Fan2006}. 

Variants of self-organizing RBF networks have also been implemented for dynamical system identification and control. Lian et al. develop a self-organizing RBF network that tunes the RBF network parameters based on an adaptation law. \cite{Lian2008a} They used this method for real-time approximation of dynamical systems. Han et al. describe a version of self-organizing RBF networks that use a growing and pruning algorithm in \cite{Han2010a}. They illustrate the effectiveness of such networks and their variants \cite{Han2018} for dynamical system identification and model predictive control. \cite{Han2011, Han2012, Qiao2012} 

Researchers have also studied the application of radial basis function networks to control problems. Some of these studies do not explicitly deal with the problem of center selection. However, the center adaptation or the kernel adaptation problems are often indirectly addressed to improve performance. In some cases, even parameter convergence is achieved. An account of common methods can be found in \cite{Sundararajan2013}. Sanner and Slotine implement Gaussian networks for direct adaptive control in \cite{Sanner1992}. The neuro-control technique discussed in \cite{Volyanskyy2008} and \cite{Kim1999} uses a fixed set of basis functions or kernel centers. On the other hand, in the controller using neural networks proposed in \cite{Patino2002}, the kernel centers are chosen such that linear independence of $\mathfrak{K}_{\bm{x}_i}$ is maintained. As per the algorithm given in \cite{Nardi2000}, the kernel parameters are chosen to approximate the nonlinear inversion error over a compact set. Reference \cite{Senanayake2018a} presents the advantages of adapting the kernel parameters and presents a theory for static as well as dynamic problems. 

An important feature of this paper is the study of how the center selection problem in RKHS embedding is related to parameter convergence in adaptive estimation. In adaptive estimation, we ordinarily use sufficient conditions, referred to as \emph{persistence of excitation} (PE) conditions, to ensure parameter convergence. \cite{Ioannou, Sastry2011, Narendra2012} The kernel center selection algorithms in the articles cited above do not take persistence of excitation into consideration. In most practical cases, the PE conditions are difficult to ensure \emph{a priori}. They often do not play a constructive role in coming up with practical algorithms. For this reason, several authors have studied adaptive estimation methods which ensure parameter convergence without PE. In \cite{Chowdhary2010}, Chowdhary and Johnson show that if the chosen regressors evaluated at measured data are linearly independent, then we get parameter convergence. Kamalapurkar et al. extended this work in \cite{Kamalapurkar2017} to relax the assumptions and developed a concurrent learning technique that implements a dynamic state-derivative estimator. Kingravi et al. in \cite{Kingravi2012a} propose a real-time regressors update algorithm that uses the regressors linear independence test. In \cite{Modares2013}, Modares et al. show that parameter convergence can be ensured by checking for linear independence of the filtered regressor. An alternative class of methods uses Gaussian processes for adaptive estimation and adaptive control. \cite{Chowdhary2012, Chowdhary2013a, Grande2013, Abdollahi2019} In these methods, the kernel centers are chosen at the points corresponding to the measured output data. An introduction to this theory with examples is given in \cite{Liu2018}.

The conventional PE condition is linked to the richness of the regressor functions that are used to represent the unknown function. In the RKHS embedding method, the modified PE conditions, studied in \cite{jia2020a, jia2020b}, are directly related to the kernel center positions in the state-space. Recent results have shown that the idea of persistence of excitation can be associated with positive limit sets contained in the state-space. We review this theory rigorously in Section \ref{sec_RKHSadapest}. This theory, along with the sufficient condition presented in \cite{Kurdila1995}, give us explicitly what sets in the state-space are persistently excited. Thus, for a particular class of RKHS adaptive estimators, we can choose kernel centers from these sets. The recent results in \cite{Guo2020Rates} establish that the accuracy of the RKHS embedding method can be shown to depend on the fill distance of samples in an uniform manifold. As the fill distance decreases to zero, the finite-dimensional approximation of function estimate converges to the infinite-dimensional function estimate. At the same time, it is also known that the condition number of the Grammian matrix that must be inverted to implement the RKHS embedding method is bounded by the minimal separation of samples that define the space of approximants. These two observations suggest that strategies to control the distribution of samples in practical simulations are needed.

In this paper, we first prove that the infinite-dimensional PE condition implies uniform convergence of the parameter error in the PE sets (Corollary \ref{cor_PEPWconv}). This proof strengthens the results in \cite{jia2020a, jia2020b} in that it provides an intuitive insight into the implications of the PE condition in the infinite-dimensional RKHS embedding method. We then discuss the theory behind approximation of the infinite-dimensional adaptive estimator and prove that choosing kernel centers in PE sets implies convergence of the function estimates at the kernel centers (Theorem \ref{thm_pwconv}). This results also strengthens the early results in \cite{Bobade2019} and provides insights that connect convergence in the RKHS norm to practical observable results in computation. Based on these results and the theory in \cite{Bobade2019, jia2020a, jia2020b, Guo2020Rates}, we develop criteria for choosing kernel centers (Subsection \ref{ssec_centCrit}). We present two kernel center selection algorithms that satisfy these criteria for certain classes of nonlinear systems. They apply to systems in which the neighborhoods of points in the positive limit sets are visited regularly by the state trajectory. In the limited literature on adaptive estimation by RKHS embedding, such algorithms are yet to be explored to the best of the authors' knowledge. The first algorithm is based on constructing centroidal Voronoi tessellations (CVT) of a polygon that surrounds the measured data. The second approach is based on Kohonen self-organizing maps. The advantages of these methods are as follows:
\begin{enumerate}
\item These algorithms choose kernel centers directly from the state-space. Such methods work for a large class of regressor functions, or types of kernels that define the RKHS.
\item We do not need explicit equations for the persistently exciting sets, which is the case in most practical applications. In the absence of such knowledge, it is hard to pick kernel centers that are evenly distributed in the persistently exciting set.
\item There are commercially available software for computing CVT and Kohonen self-organizing maps. This makes both methods simple to implement.
\end{enumerate}

We organize the sections in this paper as follows. In Section \ref{sec_RKHSadapest}, we present the theory of adaptive estimation in infinite-dimensional RKHS and basic properties of persistence of excitation. We also discuss the relation between the approximation rates and distribution of samples in the state-space. Finally, we present the criteria for center selection and illustrate the effectiveness of the criteria using an example. In Section \ref{sec_CVT}, we present the first method and theory of CVT based kernel center selection. We also prove theorems on convergence in this section. Section \ref{sec_Koho} presents the method based on Kohonen self-organizing maps. Finally, we present two examples that illustrate the effectiveness of both methods in Section \ref{sec_examples}.

%
%
\section{RKHS Embedding for Adaptive Estimation}
\label{sec_RKHSadapest}
\subsection{Reproducing Kernel Hilbert Space}
\label{ssec_RKHS}
A reproducing kernel Hilbert space $\mathcal{H}_X$ is a Hilbert space associated with a positive-definite kernel $\mathcal{K} : X \times X \to \mathbb{R}$. See \cite{Aronszajn1950, Berlinet2011} for axiomatic definitions of what constitutes an admissible kernel. The kernel satisfies two properties, (1) $\mathcal{K} (\bm{x}, \cdot) \in \mathcal{H}$ for all $\bm{x} \in X$, and (2) the reproducing property: for all $\bm{x} \in X$ and $f \in \mathcal{H}_X$, $(\mathcal{K}(\bm{x},\cdot), f)_{\mathcal{H}_X} = \mathcal{E}_{\bm{x}} f = f(\bm{x})$. Here, the notation $(\cdot, \cdot)_{\mathcal{H}_X}$ denotes the inner product associated with the Hilbert space $\mathcal{H}_X$. The term $\mathcal{E}_{\bm{x}}$ is the evaluation functional, which is a bounded linear operator. Throughout this paper, we consider RKHS generated by kernels which satisfy the condition that $\mathcal{K}(\bm{x},\bm{x}) \leq \bar{k}^2 < \infty$. This condition implies that the RKHS is continuously embedded in the space of continuous functions $C(X)$. \cite{Bobade2019} Many reproducing kernels used in practice satisfy the above condition. Given a positive-definite kernel, the RKHS $\mathcal{H}_X$ is generated by
\begin{equation*}
\mathcal{H}_X := \overline{ span \{ \mathcal{K}(\bm{x},\cdot) | \bm{x} \in X \}}.
\end{equation*}
Note that if the set $X$ is infinite-dimensional, then the RKHS it generates is also infinite-dimensional. Given a subset $\Omega \subseteq X$, we define the associated RKHS $\mathcal{H}_\Omega \subseteq \mathcal{H}_X$ by
\begin{equation*}
\mathcal{H}_\Omega := \overline{ span \{ \mathcal{K}(\bm{x},\cdot) | \bm{x} \in \Omega \}}.
\end{equation*}
The above-mentioned reproducing property endows the RKHS with a structure that makes calculations easier. A detailed list of properties of RKHS can be found in \cite{Aronszajn1950, Berlinet2011}. In this paper, we are particularly interested in the properties of projection operators that act on an RKHS. We let $P_{\Omega}$ be the $\mathcal{H}_X$ orthogonal projection operator $P_{\Omega} : \mathcal{H}_X \to \mathcal{H}_\Omega$. From Hilbert space theory, we know that the operator $P_\Omega$ decomposes the Hilbert space $\mathcal{H}_X$ into $\mathcal{H}_\Omega \bigoplus \mathcal{V}_\Omega$, where $\mathcal{V}_\Omega$ is the space of elements orthogonal to the elements of the space $\mathcal{H}_\Omega$. Since the space $\mathcal{H}_X$ is an RKHS, the reproducing property implies that for any $h \in \mathcal{V}_\Omega$, we have $h(\bm{x}) = 0$ for all $\bm{x} \in \Omega$. Another important property we use in this paper is that for any discrete finite set $\Omega_n$, the projection operator $P_{\Omega_n}$ coincides with the interpolation operator over $\Omega_n$, i.e., for all $h \in \mathcal{H}_X$, and $\bm{x} \in \Omega_n$, we have $h(\bm{x}) = (P_{\Omega_n} h) (\bm{x})$. \cite{Wendland2004}

\subsection{Adaptive Estimation in RKHS}
Consider a nonlinear system governed by the ordinary differential equation
\begin{align*}
\dot{\bm{x}}(t) = A \bm{x}(t) + B f(\bm{x}(t)),
\end{align*}
where $\bm{x}(t) \in \mathbb{R}^d$ is the state, $A \in \mathbb{R}^{d \times d}$ is a known Hurwitz matrix, $B \in \mathbb{R}^d$ is a known vector and $f : \mathbb{R}^d \to \mathbb{R}$ is the unknown (nonlinear) function. Note, if the original system equations do not contain the term $A \bm{x}(t)$, we can add and subtract a known Hurwitz matrix and redefine the unknown nonlinear function to have the form shown above. As noted in \cite{Bobade2019} and discussed in more detail there, more general systems can addressed in the analysis that follows via analogy to the model problem above.

We assume that the unknown function $f$ lives in the RKHS $\mathcal{H}_X$, where $X = \mathbb{R}^d$ is the state-space of the system. In other words, we assume that the unknown $f$ has the form $f(\cdot) = \sum_{i\in \mathbb{I}}^\infty \alpha_i \mathcal{K}_{\bm{x}_i}(\cdot)$ for some $\{\bm{x}_i \}_{i \in \mathbb{I}}$ with $\mathbb{I}$ either finite or infinite. We now define an estimator model of the form
\begin{align*}
    \dot{\hat{\bm{x}}}(t) = A \hat{\bm{x}}(t) + B \hat{f}(t,\bm{x}(t)),
\end{align*}
where $\hat{\bm{x}}(t) \in \mathbb{R}^d$ is the state estimate and $\hat{f}(t,\bm{x}(t))$ is the function estimate. For each $t$, the function estimate $\hat{f}(t)$ is an element of the space $\mathcal{H}_X$. In this paper, we assume full-state measurement. This assumption allows us to define a function estimate $\hat{f}(t)$ that depends on the actual states $\bm{x}(t)$. Note that the function estimate also explicitly depends on the time $t$. The goal of adaptive estimation is make $\hat{f}(t) \to f$ as $t \to \infty$. To achieve this, we define the rate of evolution of the function estimate by the learning law
\begin{align*}
    \dot{\hat{f}}(t) = \Gamma^{-1} (B \mathcal{E}_{\bm{x}(t)})^* P (\bm{x}(t) - \hat{\bm{x}}(t)),
\end{align*}
where $\Gamma \in \mathbb{R}$, $\Gamma > 0$. The notation $(\cdot)^*$ represents the adjoint of an operator. Additionally, the term $P$ is a symmetric positive-definite matrix in $\mathbb{R}^{d \times d}$ that solves the Lyapunov equation $A^T P + PA = - Q$, where $Q \in \mathbb{R}^{d \times d}$ is an arbitrarily chosen symmetric positive-definite matrix.

If we define the state and function errors as $\tilde{\bm{x}}(t) := \bm{x}(t) - \hat{\bm{x}}(t)$ and $\tilde{f}(t) := f - \hat{f}(t)$, the error evolution equations can be expressed as
\begin{align}
    \begin{Bmatrix}
    \dot{\tilde{\bm{x}}}(t) \\ \dot{\tilde{f}}(t)
    \end{Bmatrix}
    = 
    \underbrace{
    \begin{bmatrix}
    A & B \mathcal{E}_{\bm{x}(t)} \\
    - \Gamma^{-1} (B \mathcal{E}_{\bm{x}(t)})^* P & 0
    \end{bmatrix}
    }_{\mathbb{A}(t)}
    \begin{Bmatrix}
    \tilde{\bm{x}}(t) \\ \tilde{f}(t)
    \end{Bmatrix}.
    \label{eq_errEst}
\end{align}
Note, in the above error equation, the term $\mathbb{A}(t)$ is a uniformly bounded linear operator, and the states $\begin{Bmatrix} \tilde{\bm{x}}(t) & \tilde{f}(t) \end{Bmatrix}^T$ evolve in the infinite-dimensional space $\mathbb{R}^d \times \mathcal{H}_X$. 

Standard stability analysis using the Lyapunov's theorem and Barbalat's lemma shows that the norm of the state error $\|\tilde{\bm{x}}(t)\|_{\mathbb{R}^d} \to 0$ as $t \to \infty$. \cite{Bobade2019, jia2020a, jia2020b}

%
\subsection{Parameter Convergence, PE and Positive Limit Sets}
\label{ssec_paraConv}
As mentioned earlier, persistence of excitation (PE) conditions are used to prove convergence of the function estimate to the actual function. Two different definitions of PE in RKHS are available in the recent literature on RKHS embedding methods. \cite{jia2020a, jia2020b} They are as follows.

\begin{definition}
\label{def_PE1}
(\textbf{PE. $\bm{1}$}) The trajectory $\bm{x}: t \mapsto \bm{x}(t) \in \mathbb{R}^d$ persistently excites the indexing set $\Omega$ and the RKHS $\mathcal{H}_\Omega$ provided there exist positive constants $T_0, \gamma, \delta,$ and $\Delta$, such that for each $t \geq T_0$ and any $g \in \mathcal{H}_X$, there exists $s \in [t,t+\Delta]$ such that
\begin{align*}
    \left| \int_{s}^{s+\delta} \mathcal{E}_{\bm{x}(\tau)} g d \tau \right| \geq \gamma \| P_\Omega g \|_{\mathcal{H}_X} > 0.
\end{align*}
\end{definition}

\begin{definition}
\label{def_PE2}
(\textbf{PE. $\bm{2}$}) The trajectory $\bm{x}: t \mapsto \bm{x}(t) \in \mathbb{R}^d$ persistently excites the indexing set $\Omega$ and the RKHS $\mathcal{H}_\Omega$ provided there exist positive constants $T_0$, $\gamma$, and $\Delta$ such that 
\begin{align*}
    \int_t^{t+\Delta} \left( \mathcal{E}^*_{\bm{x}(\tau)} \mathcal{E}_{\bm{x}(\tau)}g,g \right)_{\mathcal{H}_X} d \tau \geq \gamma \| P_\Omega g \|_{\mathcal{H}_X}^2 > 0
\end{align*}
for all $t \geq T_0$ and any $g \in \mathcal{H}_X$.
\end{definition}

Note that the PE condition given in Definition \ref{def_PE2} structurally resembles the classical PE conditions defined using regressors in finite-dimensional spaces. \cite{Ioannou, Sastry2011, Narendra2012} Recall that the term $P_\Omega$ in the above definitions is the orthogonal projection operator that maps elements from $\mathcal{H}_X$ to $\mathcal{H}_\Omega$. The following theorem from \cite{jia2020a, jia2020b} shows how these two PE conditions are related, and the PE condition in Definition \ref{def_PE1} implies parameter convergence. Note that the notion of parameter convergence in the infinite-dimensional case is given with respect to PE condition in Definition $\ref{def_PE1}$ only.

\begin{theorem}
\label{thm_PEInfConv}
The PE condition in Definition PE. \ref{def_PE1} implies the one in Definition PE. \ref{def_PE2}. Further, if $X = \Omega$ is a discrete finite set, the state trajectory $t \mapsto \bm{x}(t)$ is uniformly continuous and maps to a compact set, and the family of functions defined by $\{ g(\bm{x}(\cdot)): t \mapsto g(\bm{x}(t)) | g \in \mathcal{H}_X, \| g \|=1 \}$ is uniformly equicontinuous, then the PE condition in Definition PE. \ref{def_PE2} implies the one in Definition PE. \ref{def_PE1}.

Furthermore, if the trajectory $\bm{x} : t \mapsto \bm{x}(t)$ persistently excites the RKHS $\mathcal{H}_\Omega$ in the sense of Definition PE. \ref{def_PE1}. Then
\begin{align*}
    \lim_{t \to \infty} \| \tilde{\bm{x}}(t) \| = 0, \hspace{0.75in} \lim_{t \to \infty} \| P_\Omega \tilde{f}(t) \|_{\mathcal{H}_X} = 0.
\end{align*}
\end{theorem}

We can view the term $P_\Omega \tilde{f}(t)$ as an element of the space $\mathcal{H}_\Omega$. Thus, the above statement implies that $P_\Omega \tilde{f}(t)$ converges to the zero element in the $\mathcal{H}_\Omega$ space. However, this statement does not imply the convergence or even the existence of the limit of $\tilde{f}(t) \in \mathcal{H}_X$. 

The statement $\lim_{t \to \infty} \| P_\Omega \tilde{f}(t) \|_{\mathcal{H}_X} = 0$ is hard to interpret intuitively. The following corollary of the above theorem gives us the intuition about where the convergence is achieved. 

\begin{corollary}
\label{cor_PEPWconv}
If the trajectory $\bm{x} : t \mapsto \bm{x}(t)$ persistently excites the set $\Omega$ and the RKHS $\mathcal{H}_\Omega$ in the sense of Definition PE. \ref{def_PE1}, then $ \hat{f}(t) $ converges uniformly to $f$ on the set $\Omega$ as $t \to \infty$.
\end{corollary}

\begin{proof}
Suppose the projection operator $P_\Omega$ decomposes the function $\tilde{f}(t)$ into $\tilde{f}(t) = P_{\Omega} \tilde{f}(t) + v(t)$, where $P_{\Omega} (\tilde{f}(t)) \in \mathcal{H}_{\Omega}$ and $v(t) \in \mathcal{V}_\Omega$. Since $v(t,\bm{x}) = 0$ for all $\bm{x} \in \Omega$, we have $\tilde{f}(t) = P_{\Omega} \tilde{f}(t,\bm{x})$. Thus, for all $\bm{x} \in \Omega$, we have
\begin{align*}
    |\tilde{f}(t)| = |P_{\Omega} \tilde{f}(t,\bm{x})| = |\mathcal{E}_{\bm{x}} P_{\Omega} \tilde{f}(t)| \leq \| \mathcal{E}_{\bm{x}} \| \| P_{\Omega} \tilde{f}(t) \|_{\mathcal{H}_X}.
\end{align*}
But we have assumed in this paper that the kernel $\mathcal{K}$ that induces $\mathcal{H}_X$ satisfies $\mathcal{K}(\bm{x},\bm{x}) \leq \bar{k}^2 < \infty$ for all $\bm{x} \in X$. Since the evaluation functional is consequently uniformly bounded, the above inequality holds for all $\bm{x} \in \Omega$. Taking the limit $t \to \infty$ and using Theorem \ref{thm_PEInfConv} gives us the desired result.
\end{proof}

The above corollary clearly shows that, if the PE condition holds and \emph{the kernel satisfies $\mathcal{K}(\bm{x},\bm{x}) \leq \bar{k}^2 < \infty$}, then $\hat{f}(t,\bm{x}) \to f(\bm{x})$ for all $\bm{x} \in \Omega$. Generally, we would prefer the whole space to be persistently exciting, i.e. $\Omega = X$. However, this is not the case in most practical applications. Furthermore, the above PE definitions are hard to understand intuitively and difficult, if not impossible, to verify in real applications. The following theorem from \cite{Kurdila2019PE} shows us exactly where to look for persistently exciting sets in the state-space. The theorem assumes that the RKHS space separates closed sets.

\begin{definition}
\label{def_sepset}
We say the RKHS $\mathcal{H}_X$ separates a set $A \subseteq X$ if for each $\bm{b} \notin A$, there is a function $f \in \mathcal{H}_X$ such that $f(\bm{a}) = 0$ for all $\bm{a} \in A$ and $f(\bm{b}) \neq 0$.
\end{definition}

\begin{condition}
\label{cond_sepSets}
The RKHS $\mathcal{H}_X$ separates closed sets.
\end{condition}

The RKHS generated by the Gaussian kernel, which is extensively used for RKHS based adaptive estimation and machine learning, does not satisfy the above condition for all closed sets. A detailed account for RKHS that separate closed sets can be found in \cite{DeVito2012}. In this paper, we use the Sobolev-Matern kernels, which satisfy the above condition.

\begin{theorem}
\label{thm_PElimitset}
Let $\mathcal{H}_X$ be the RKHS of functions over $X$ and suppose that this RKHS includes a rich family of bump functions. If the PE condition in Definition PE. \ref{def_PE2} holds for $\Omega$, then $\Omega \subseteq \omega^+(\bm{x}_0)$, the positive limit set corresponding to the initial condition $\bm{x}_0$.
\end{theorem}

When the RKHS satisfies Condition \ref{cond_sepSets}, this theorem gives us a necessary condition for a set to be persistently excited. While designing a adaptive estimator, this necessary condition can tell us where to look for persistently excited sets in the state-space.

\subsection{Approximations, Convergence Rates and Sufficient Condition}
\label{ssec_RKHSapprox}
For practical implementation, we approximate the infinite-dimensional adaptive estimator equations given in the previous subsection. Let $\{\Omega_n\}_{n \in \mathbb{N}}$ be a finite nested sequence of subsets of $\Omega$, Further, let $\{ \mathcal{H}_{\Omega_n} \}_{n \in \mathbb{N}}$ be the corresponding subspaces of $\mathcal{H}_X$ generated by the finite sets $\Omega_n$. Now, define $P_{\Omega_n}$ as the orthogonal projection operator from $\mathcal{H}_X$ to the subspace $\mathcal{H}_{\Omega_n}$ such that $\lim_{n \to \infty} P_{\Omega_n} f = f$ for all $f \in \mathcal{H}_X$. With this definition of approximation, we write the finite-dimensional adaptive estimator model and the learning law as

\begin{align*}
    \dot{\hat{\bm{x}}}_n(t) &= A \hat{\bm{x}}_n(t) + B \mathcal{E}_{\bm{x}(t)} \Pi_n^* \hat{f}_n(t), 
    \\
    \dot{\hat{f}}_n(t) &= \Gamma^{-1} \left( B \mathcal{E}_{\bm{x}(t)} \Pi_n^* \right)^* P \tilde{\bm{x}}_n(t)
\end{align*}
with $ \tilde{\bm{x}}_n:= \bm{x} - \hat{\bm{x}}_n$. Since the RKHS $\mathcal{H}_{\Omega_n}$ is finite-dimensional, the basis of $\mathcal{H}_{\Omega_n}$ is the set $\{ \mathfrak{K}_{\bm{x}_i}| \bm{x}_i \in \Omega_n\}$. We now note that the finite-dimensional function estimate $\hat{f}_n(t)$ has the form $\hat{f}_n (t) := \sum_{i=1}^n \hat{\alpha}_i (t) \mathfrak{K}_{\bm{x}_i}$. Using the reproducing property of the kernel, we rewrite the above finite-dimensional learning law as

\begin{align}
    \dot{\hat{\bm{\alpha}}}(t)= \mathbb{K}^{-1} \bm{\Gamma}^{-1} \bm{\mathcal{K}}(\bm{x}_{c},\bm{x}(t)) B^* P \tilde{\bm{x}}_n(t),
    \label{eq_FinLLaw}
\end{align}
where $\hat{\bm{\alpha}}(t):= \{ \hat{\alpha}_1(t),\ldots,\hat{\alpha}_n(t) \}^T$, $\mathbb{K}$ is the symmetric positive definite Grammian matrix whose $ij^{th}$ element is defined as $\mathbb{K}_{ij} := \mathcal{K}(\bm{x}_i,\bm{x}_j)$, $\bm{\Gamma}:= \Gamma \mathbb{I}_{n}$ is the gain matrix, and 
$$\bm{\mathcal{K}}(\bm{x}_{c},\bm{x}(t)):= \begin{Bmatrix} \mathcal{K}(\bm{x}_1,\bm{x}(t)),\ldots,\mathcal{K}(\bm{x}_n,\bm{x}(t)) \end{Bmatrix}^T.$$

The new learning law defines the rate of evolution of the coefficients, as opposed to the old learning law which defines the rate of evolution of the function $\hat{f}_n(t)$. This step is essential for implementation purposes. We refer the reader to \cite{Paruchuri2020Piezo} for the intermediate steps involved in the derivation.

Note, the PE condition implies the convergence of the infinite-dimensional function estimate $\hat{f}(t)$ to $f$. It does not imply anything about the convergence of the approximation of the function estimate $\hat{f}_n(t)$ to $f$. On the other hand, the following theorem, proved in \cite{Bobade2019}, shows that the term $\hat{f}_n(t)$ to $\hat{f}(t)$ as $n \to \infty$.

\begin{theorem}
Suppose that $\bm{x} \in C([0,T],\mathbb{R}^d)$ and that the embedding $i : \mathcal{H}_X \hookrightarrow C(\Omega)$ is uniform in the sense that 
\begin{align*}
    \|f\|_{C(\Omega)} \equiv \|if\|_{C(\Omega)} \leq C \|f\|_{\mathcal{H}_X}.
\end{align*}
Then for any $T>0$ and $t \in [0,T]$,
\begin{align*}
    \| \hat{\bm{x}} - \hat{\bm{x}}_n \|_{C([0,T];\mathbb{R}^d)} &\to 0, \\
    \|\hat{f}(t) - \hat{f}_n(t) \|_{C([0,T];\mathbb{R}^d)} &\to 0,
\end{align*}
as $n \to \infty$.
\end{theorem}

Thus, as we choose denser finite discrete sets in $\Omega$, the approximation of the function estimate $\hat{f}_n(t)$ gets closer to the function estimate $\hat{f}(t)$, which in turn converges to the actual function $f$ as $t \to \infty$ if the PE condition holds. The above theorem does not explicitly tell us how to choose the set $\Omega_n \subseteq \Omega$. However, when the set $\Omega$ is a compact smooth Riemannian manifold embedded in $\mathbb{R}^d$ with metric $d$, the rate at which $\hat{f}_n(t)$ converges to the $\hat{f}(t)$ depends on how the elements of the set $\Omega_n$ are distributed in the set $\Omega$. This distribution is defined in terms of the \emph{fill distance} 
\begin{equation*}
    h_{\Omega_n,\Omega} := \sup_{x\in\Omega}\min_{\xi_i\in\Omega_n} d(x,\xi_i).
\end{equation*}

\begin{theorem}
\label{thm_approxrates}
Let $\Omega \subseteq X:=\mathbb{R}^d$ be a $k$-dimensional smooth manifold, and let the native space $\mathcal{H}_X$ be continuously embedded in a Sobolev space $W^{\tau,2}(X)$
with $\tau>d/2$, so that $\|f\|_{W^{\tau,2}(\mathbb{R}^d)} \lesssim \|f\|_{\mathcal{H}_X}$. Define $s=\tau-(d-k)/2$ and let $0\leq \mu \leq \lceil s\rceil -1$. Then there is a constant $h_\Omega$ such that if $h_{\Omega_n,\Omega}\leq h_\Omega$, then for all $f\in \bm{R}_\Omega (\mathcal{H}_X)$ we have 
\begin{equation*}
  \|(I-P_{\Omega_n})\hat{f}(t)\|_{W^{\mu,2}(\Omega)} \lesssim h_{\Omega_n,\Omega}^{s-\mu}\|\hat{f}(t)\|_{\bm{R}_\Omega (\mathcal{H}_X)}.  
\end{equation*}
\end{theorem}

In the above theorem, the notation $\bm{R}_\Omega (\mathcal{H}_X)$ represents the restriction of the space $\mathcal{H}_X$ to the set $\Omega$, and the notation $a \lesssim b$ implies that there exists a positive constant $c$ such that $a \leq c b$. This theorem requires a lot of technical details and we direct interested readers to \cite{Guo2020Rates} for the detailed explanation of the rigorous theory and proofs. In this paper, we are interested in the implications of the theorem. The theorem states that the fill distance $h_{\Omega_n,\Omega}$ defines the rate at which the norm of the error $\hat{f}(t) - \hat{f}_n(t)$ converges to zero.

\subsubsection{Sufficient Condition}

In all the discussion above, we assume that we have knowledge of the persistently excited set $\Omega$. In most practical cases, it is impossible to determine this set. However, there is a much more practical and intuitive way for selecting the set $\Omega_n$ when the RKHS is generated by a radial basis kernel.

\begin{condition}
\label{cond_rbf}
The RKHS is generated by a radial basis kernel.
\end{condition}

\begin{theorem}
\label{thm_suff}
Let $\epsilon < \frac{1}{2} \min_{i \neq j} \| \bm{x}_i - \bm{x}_j \|$, where $\bm{x}_i$ and $\bm{x}_j$ are the kernel centers $\{\bm{x}_1, \ldots, \bm{x}_n \} \subseteq \omega^+(\bm{x}_0)$. For every $t_0 \geq 0$ and $\delta > 0$, define
\begin{align*}
    I_i := I_{i,\epsilon,\delta} := \{ t \in [t_0, t_0 + \delta] : \| \bm{x}(t) - \bm{x}_i \| \leq \epsilon \}.
\end{align*}
If there exists a $\delta = \delta(\epsilon)$ such that the measure of $I_i$ is bounded below by a positive constant that is independent of $t_0$ and the kernel center $\bm{x}_i$, and if the measure of $[t_0,t_0+\delta]$ is less than or equal to $\delta$, then the space $\mathcal{H}_n$ is persistently exciting in the sense of PE \ref{def_PE2}.
\label{thm_sc}
\end{theorem}
Intuitively, the above theorem states that the neighborhoods of the points in the finite PE set $\Omega_n$ are visited by the state trajectory infinitely many times. The proof of this theorem is given in \cite{Kurdila1995}, where the theorem is stated for a specific class of radial basis functions. However, the radial basis functions used in this paper, and the ones used most commonly satisfy these conditions. Furthermore, the original theorem in \cite{Kurdila1995} stipulates additional conditions on $\epsilon$. However, when the kernel centers are contained in the positive limit set $\omega^+(\bm{x}_0)$, there always exists an $\epsilon$ such that these additional conditions are satisfied. Note that the sufficient condition implies PE \ref{def_PE2}. However, when the hypotheses of Theorem \ref{thm_PEInfConv} hold, we can conclude that the sufficient condition given in Theorem \ref{thm_sc} implies PE \ref{def_PE1}. While implementing the adaptive estimator, if the actual function $f \in \mathcal{H}_X$, where $X$ is an infinite set, the sufficient condition given in Theorem \ref{thm_sc} only implies ultimate boundedness of the function estimate instead of convergence, in particular when we use the dead zone gradient law.

%
\subsection{Center Selection Problem and Example}
\label{ssec_cenProb}
In the last section, we made no assumption about the space in which function estimate $\hat{f}(t)$ lives. The function estimate $\hat{f}(t)$ can live in $\mathcal{H}_X$ and is not restricted to $\mathcal{H}_\Omega$. This leads us to ask the question of why it is necessary for the kernel centers (elements of the set $\Omega_n$) to be contained in the set $\Omega$. It is indeed possible to approximate the function $\hat{f}(t)$ using kernel centers that are outside of the set $\Omega$. However, if the centers are contained in the set $\Omega$, the function estimate will converge to the actual function values at those centers. Before we take a look at the next theorem, note that the basis of the space $\mathcal{H}_{\Omega_n}$ is the set $\{ \mathfrak{K}_{\bm{x}_i}| \bm{x}_i \in \Omega_n\}$. This implies that the functions $P_{\Omega_n} f$ and $\hat{f}_n(t)$ have the form $P_{\Omega_n} f = \sum_{i=1}^n \alpha_i \mathfrak{K}_{\bm{x}_i}$ and $\hat{f}_n(t) = \sum_{i=1}^n \hat{\alpha}_i(t) \mathfrak{K}_{\bm{x}_i}$.

\begin{theorem}
\label{thm_pwconv}
Suppose the set $\Omega$ is persistently exciting, and the set $\Omega_n \subseteq \Omega$. Then $\lim_{t \to \infty} \hat{f}_n(t,\bm{x}_i) = f(\bm{x}_i)$ for all $\bm{x} \in \Omega_n$ and $i \in \{ 1,\ldots,n \}$. Furthermore, for $i \in \{ 1,\ldots,n \}$, $\hat{\alpha}_i(t) \to \alpha_i$ as $t \to \infty$.
\end{theorem}

\begin{proof}
Recall that $\hat{f}_n(t) := P_{\Omega_n} \hat{f}(t)$, where $P_{\Omega_n} : \mathcal{H}_X \to \mathcal{H}_{\Omega_n}$. The set $\Omega_n$ is discrete and finite. In RKHS, the projection operator from infinite-dimensional space to a finite-dimensional space coincides with the interpolation operator. In other words, for a given $t$, we have $\hat{f}_n(t,\bm{x}_i) = \hat{f}(t,\bm{x}_i)$ for all $\bm{x}_i \in \Omega_n$. From Corollary \ref{cor_PEPWconv}, we have $\lim_{t \to \infty} \hat{f}_n(t,\bm{x}_i) = f(\bm{x}_i)$ for all $\bm{x}_i \in \Omega_n$. This in turn implies that, for $i \in \{ 1,\ldots,n \}$, the coefficients $\hat{\alpha}_i(t)$ converge to $\alpha_i$ as $t \to \infty$ since  the set $\{ \mathfrak{K}_{\bm{x}_i}| \bm{x}_i \in \Omega_n\}$ forms the basis of the space $\mathcal{H}_{\Omega_n}$.
\end{proof}

The  above theorem shows that selecting kernel centers in the PE set $\Omega$ will result in the approximated function estimate $\hat{f}_n(t)$ approaching the actual function value at the kernel centers. In addition to this fact, the theory on approximation rates holds only when the kernel centers are contained in the set $\Omega$. This makes it advantageous to choose $\Omega_n \subseteq \Omega$. The following example helps us understand what happens when the kernel center is not exactly in the persistently excited set. The example considers the case where $\Omega$ is a singleton set. The analysis for more general PE sets is analogous to the one given below.

\begin{example}
Suppose the persistently excited set $\Omega = \{ \bm{\xi} \}$. Suppose the kernel center is at $\Omega_n = \{ \hat{\bm{\xi}} \}$. According to Corollary \ref{cor_PEPWconv}, given $\epsilon > 0$, there exists a $T_0$ such that for any $t > T_0$, $|f(\bm{\xi}) - \hat{f}(t,\bm{\xi})| < \epsilon$. Suppose we stop the adaptive estimator at $T > T_0$. We know that by the properties of RKHS, $\hat{f}(T,\hat{\bm{\xi}}) = \hat{f}_n(T,\hat{\bm{\xi}})$. Since $\hat{f}$ and $\hat{f}_n$ are continuous, given $\epsilon > 0$, there exists $\delta$ such that if $\|\bm{\xi} - \hat{\bm{\xi}} \| < \delta$, then $| \hat{f}(T,\bm{\xi}) - \hat{f}(T,\hat{\bm{\xi}}) | < \epsilon$ and $| \hat{f}_n(T,\bm{\xi}) - \hat{f}_n(T,\hat{\bm{\xi}}) | < \epsilon$. Thus, we conclude that if $\|\bm{\xi} - \hat{\bm{\xi}} \| < \delta$, then $|f(\bm{\xi}) - \hat{f}_n(T,\bm{\xi})| < 3 \epsilon$. Note, as $\hat{\bm{\xi}} \to \bm{\xi}$, $|f(\bm{\xi}) - \hat{f}_n(T,\bm{\xi})|$ approaches a value that is strictly less than $\epsilon$. 
\end{example}

\subsection{Center Selection Criteria}
\label{ssec_centCrit}
Based on the theory presented in the previous subsections, we list the following criteria for choosing the kernel centers.
\begin{enumerate}[label = (C{\arabic*})]
    \item The kernel centers should be contained in or be as close as possible to the positive limit set (based on Theorem \ref{thm_PElimitset} provided Condition \ref{cond_sepSets} holds).
    \item The kernel centers should be evenly distributed when possible. There are two reasons for selecting this criteria. 
    \begin{enumerate}[label=(\roman*)]
        \item The linear dependency of the kernels will be high if the centers are placed too close to each other. This will increase the condition number of the Grammian matrix in Equation \ref{eq_FinLLaw}.
        \item On the other hand, if the centers are too far apart, the fill distance increases, which in turn reduces the approximation rates based on Theorem \ref{thm_approxrates}.
    \end{enumerate}
    \item The neighborhood of the centers should be visited by the state trajectory regularly. This is to satisfy the sufficient condition for PE based on Theorem \ref{thm_suff} provided Condition \ref{cond_rbf} holds.
\end{enumerate}
\textit{Note:} The above listed criteria assumes knowledge of the positive limit set and the state-trajectory.

\subsection{Example: The case when we have a priori knowledge of positive limit set}
\label{ssec_sampex}
\begin{figure}
\centering
\includegraphics[scale = 0.5]{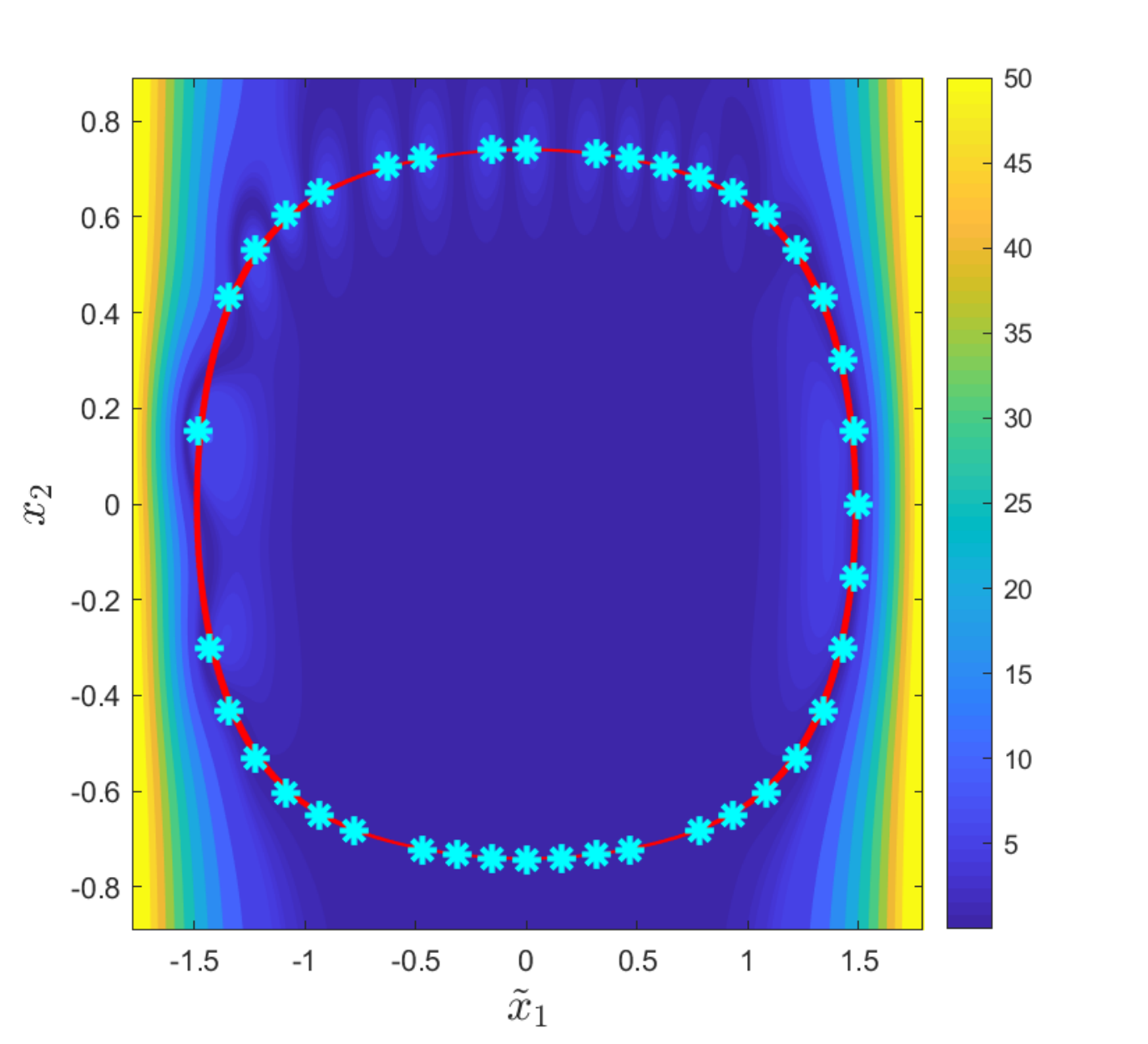}
\caption{Random Centers - Pointwise error $|f(\bm{x}) - \hat{f}_n(T,\bm{x})|$. The marker $*$ and the red line represent the kernel centers and the limit set, respectively.}
\label{fig_Ex1Rand}
\end{figure}

\begin{figure}
\centering
\includegraphics[scale = 0.5]{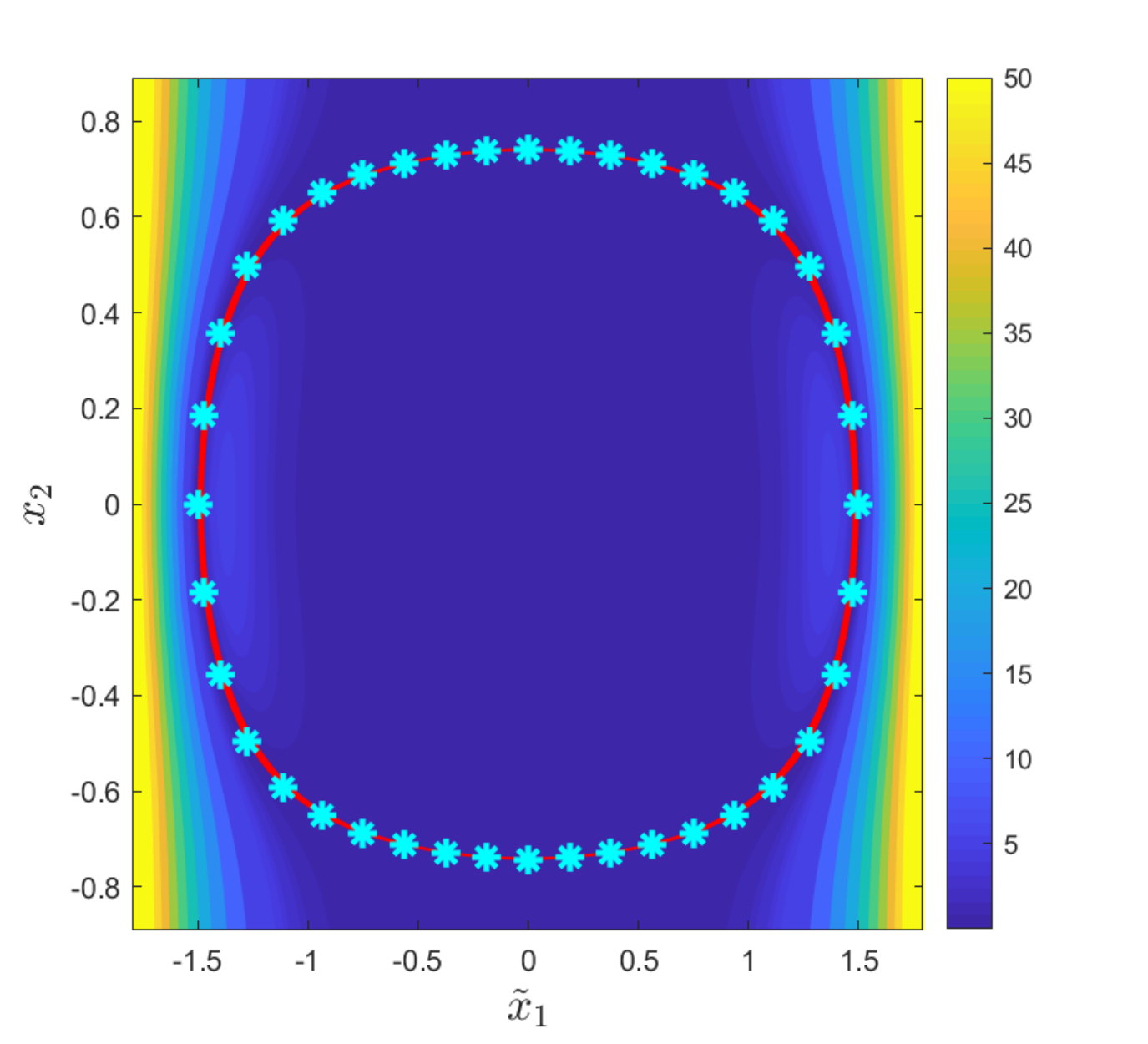}
\caption{Uniform Centers - Pointwise error $|f(\bm{x}) - \hat{f}_n(T,\bm{x})|$. The marker $*$ and the red line represent the kernel centers and the limit set, respectively.}
\label{fig_Ex1Unif}
\end{figure}

We test the above listed criteria on a simple practical example. We consider a nonlinear single-mode undamped piezoelectric oscillator with no input to test the above criteria. The governing equations have the form

\begin{align}
    \begin{Bmatrix}
    \dot{x}_1 \\ \dot{x}_2
    \end{Bmatrix}
    &=
    \underbrace{\begin{bmatrix}
    0 & 1 \\
    -\frac{\hat{K}}{M} & -\frac{C}{M}
    \end{bmatrix}}_{A}
    \begin{Bmatrix}
    x_1 \\ x_2
    \end{Bmatrix}
    +
    \underbrace{\begin{Bmatrix}
    0 \\ -\frac{P}{M}
    \end{Bmatrix}}_{\mathcal{B}} \underbrace{\ddot{\mathrm{z}}(t)}_{\mathrm{u}(t)}
    \notag
    \\
    & \hspace{0.5in}
    +
    \underbrace{\begin{Bmatrix}
    0 \\ 1
    \end{Bmatrix}}_{B}
    \underbrace{ 
    \left(
    - \frac{\hat{K}_{N_1}}{M} x_1^3(t) - \frac{\hat{K}_{N_2}}{M} x_1^5(t) 
    \right)
    }_{f(\bm{x}(t))},
    \label{eq_piezomodel}
\end{align}
where $M, \hat{K}, C, P$ are the modal mass, modal stiffness, modal damping, and modal input contribution term of the piezoelectric oscillator. The variables $\hat{K}_{N_1}, \hat{K}_{N_2}$ are the nonlinear stiffness terms. The terms $x_1$, $x_2$ and $\mathrm{z}$ are the modal displacement, modal velocity and base displacement of the oscillator, respectively. The steps involved in deriving the above governing equations can be found in \cite{Paruchuri2020Piezo}. Typically, the magnitudes of the velocity and displacement values are not of the same order. In such cases, we have to use kernels that are skewed in a particular direction. Alternatively, we scale one of the states as $x_1 = S \tilde{x}_1$, where $S$ is a positive constant. Note, after scaling, $\bm{x}(t) := \{ \tilde{x}_1(t),x_2(t) \}^T$. In our simulations, we choose $M = 0.9745$, $\hat{K} = 329.9006$, $\hat{K}_{N_1} = -1.2901 \times 10^{5}$ and $\hat{K}_{N_2} = 1.2053 \times 10^{9}$. For the undamped, no input case, i.e., $C = 0 $ and $P = 0$, the total energy is conserved. In other words, the trajectory is always contained in the limit set $\omega^+(\bm{x}_0)$, where $\bm{x}_0 \in \mathbb{R}^2$ is the initial condition. Note that any arbitrary discrete finite set in $\omega^+(\bm{x}_0)$ is visited by the state trajectory infinitely many times.

\begin{figure}
\centering
\includegraphics[scale = 0.5]{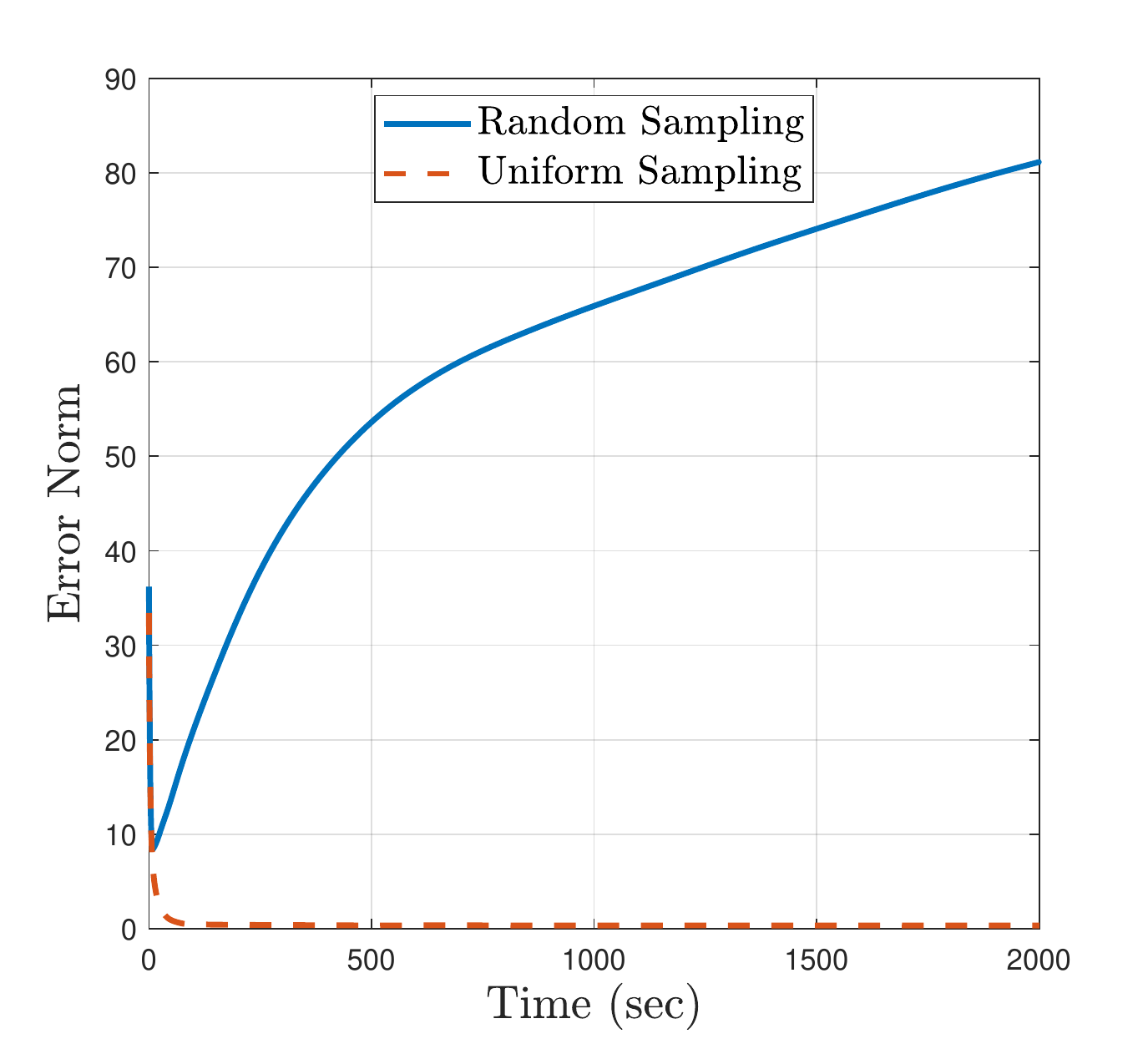}
\caption{Variation of $\| \bm{\alpha} - \hat{\bm{\alpha}}(t) \|_{\mathbb{R}^n}$ with time.}
\label{fig_Ex1ErrNorm}
\end{figure}

Since we have a priori knowledge of the limit set $\omega^+(\bm{x}_0)$ for a given initial condition, we choose kernel centers in the set $\Omega$ and integrate the equations
\begin{align*}
    \dot{\hat{\bm{x}}}_n(t) &= A \hat{\bm{x}}_n(t) + B \hat{\bm{\alpha}}^T(t) \bm{\mathcal{K}}(\bm{x}_c, \bm{x}(t)),
    \\
    \dot{\hat{\bm{\alpha}}}(t) &= \mathbb{K}^{-1} \bm{\Gamma}^{-1} \bm{\mathcal{K}}(\bm{x}_{c},\bm{x}(t)) B^* P \tilde{\bm{x}}_n(t)
\end{align*}
over the interval $[0,T]$ for some $T > 0$. In all our simulations, we use the Sobolev-Matern $3,2$ kernel, which has the form
\begin{align*}
    \mathcal{K}_{3,2}(\bm{x},\bm{y}) &= \left(1 + \frac{\sqrt{3} \|\bm{x}-\bm{y}\|}{l}\right)\exp{\left(-\frac{\sqrt{3} \|\bm{x}-\bm{y}\|}{l}\right)},
\end{align*}
where $l$ is the scaling factor of length. \cite{Rasmussen2003} 

To analyze the above-listed criteria's effectiveness, we tested the adaptive estimator with a random and uniform collection of kernel centers. We set $S = 0.02$, $l = 0.2$, $\Gamma = 0.001$ and $n = 40$. The states and the parameters are initialized at $\bm{x}_0 = \{ 1.5, 0 \}^T$ and $\alpha_i(0) = 1$ for $i = 1,\ldots,n$, respectively. For uniform kernel center selection, we first calculate the distance between two adjacent kernel centers $l_n$ when they are distributed uniformly in the positive limit set. Since we know the exact equation of the positive limit set, \cite{jia2020a} we can calculate the total length and hence the length of the arc between two adjacent kernel centers. Given a kernel center, we choose the adjacent kernel center at a distance $l_n$. We repeat this procedure until we choose the required number of kernel centers that are distributed uniformly in the positive limit set. For choosing the kernel centers for the random case, we first ran the uniform center selection algorithm for $n = 48$ case, and then used the MATLAB function \mcode{randperm} to select $n = 40$ kernel centers randomly. Note that the MATLAB function \mcode{randperm} uses a uniform pseudorandom number generator algorithm.

Figures \ref{fig_Ex1Rand} and \ref{fig_Ex1Unif} show the pointwise error $|f(\bm{x}) - \hat{f}_n(T,\bm{x})|$ after running the adaptive estimator for $T = 2000$ seconds for a paritcular case of random and uniform selection of kernel centers. It is clear from the figures that the pointwise error is low in the case of uniform sampling. Figure \ref{fig_Ex1ErrNorm} shows how the norm $\| \bm{\alpha} - \hat{\bm{\alpha}}(t) \|_{\mathbb{R}^n}$ varies with time $t$ for both the random and uniform center selection methods. From Theorem \ref{thm_pwconv}, we know that $\hat{\bm{\alpha}}(t) \to \bm{\alpha}$, where $\bm{\alpha} = \{\alpha_1,\ldots,\alpha_n\}^T$ and $\hat{\bm{\alpha}}(t) = \{\hat{\alpha}_1(t),\ldots,\hat{\alpha}_n(t)\}^T$. It is clear from Figure \ref{fig_Ex1ErrNorm} that the coefficient error norm converges rapidly to zero for the uniform centers case. For the random centers case, the error norm does not even start converging in the first $2000$ seconds.

In the above problem, it is assumed that we have an explicit equation for the positive limit set $\omega^+(\bm{x}_0)$ for a given initial condition $\bm{x}_0$. Furthermore, the state trajectory is contained in the set $\omega^+(\bm{x}_0)$. This makes it possible to choose kernel centers that are uniformly distributed. In most practical examples, we cannot derive an explicit expression for the set $\omega^+(\bm{x}_0)$. We only have samples of the state-trajectory that is contained in or converges to the positive limit set $\omega^+(\bm{x}_0)$. In the following two sections, we present kernel center selection methods that can be implemented when we do not have explicit knowledge of the positive limit set or when the state trajectory is not contained in the positive limit set. Both methods are applicable to systems for which the state trajectory visits the neighborhoods of all the points in the positive limit set $\omega^+(\bm{x}_0)$. We next consider algorithms that do not rely on a priori knowledge of the positive limit set $\omega^+(\bm{x}_0)$.

%
\section{Method 1: Based on CVT and Lloyd's Algorithm}
\label{sec_CVT}
The first method we propose is based on building centroidal Voronoi tessellations (CVT) around the positive limit set. This method relies on samples taken in the positive limit set. We implement this approach for systems where the state-trajectory is contained in the positive limit set or converges to the same in finite time. We assume that there is a dense sampling $\Xi$ of the positive limit set, i.e. $\overline{\Xi} = \omega^+(\bm{x}_0)$. Let $\{ \Xi_m \}_{m=1}^\infty$ be a sequence of finite subsets of $\Xi$ such that $\Xi_m \subset \Xi_{m+1}$ for all $m \in \mathbb{N}$ and $\cup_{m=1}^\infty \Xi_m = \Xi$, where $\Xi_m = \{ \bm{\xi}_1,\ldots,\bm{\xi}_{q_m} \}$. The term $q_m$ represents the number of samples in the set $\Xi_m$. Given a set of samples $\Xi_m$, we construct a region $Q_m$ that is assumed to enclose the positive limit set. Before we go into the details of implementation, let us take a look at the theory behind Voronoi partitions.
%
%
\subsection{Voronoi Partition}
Suppose the state-space $X$ is endowed with the metric $d(\cdot,\cdot)$. In this paper, we use the Euclidean metric. Let $Q_m \subseteq X$ be a convex polytope and let $P_m = \{\bm{p}_{m,1},\ldots,\bm{p}_{m,n_m} \}$ be a set of $n_m$ points. The Voronoi partition $\mathcal{V}(P_m)$ generated by the set of points $P_m$ is the collection of $n_m$ polytopes, $P_{m,1}, \ldots, P_{m,n_m}$, defined by
\begin{align*}
P_{m,i} &= \left \{ \bm{x} \in Q_m \mid d(\bm{x},\bm{x}_i) \leq d(\bm{x},\bm{x}_j), \right. \\
& \hspace{1in} \left. \text{ for } j = 1,\ldots, n_m, j \neq i \right \}
\end{align*}
for $i = 1,\ldots,n_m$. An edge of the polytope $P_{m,i}$ is the region $P_{m,i} \cap P_{m,j}$ or $P_{m,i} \cap \partial Q_m$ for some $j \neq i$. We say that two polytopes are adjacent when they share a common edge. The notation $\partial Q_m$ denotes the boundary of the region $Q_m$. We use the notation $\mathbb{E}(\mathcal{V}(P_m),Q_m)$ to denote the union of all edges of the polytopes in $\mathcal{V}(P_m)$. If $R \subseteq Q_m$, then $\mathbb{E}(\mathcal{V}(P_m),R) = \mathbb{E}(\mathcal{V}(P_m),Q_m) \cap R$. A particular class of Voronoi partitions are the \emph{centroidal Voronoi partitions} or \emph{centroidal Voronoi tessellations}, where each point generating the polytope is also its centroid. We use the notation $C_{P_{m,j}}$ to denote the centroid that generates the polytope $P_{m,j}$. Note, given a region $Y \subseteq X$ in the state-space, its centroid $C_{Y}$ is defined as
\begin{align*}
C_Y = \frac{1}{M_Y} \int_{Y} \bm{y} \rho (\bm{y}) d\bm{y},
\end{align*}
where $M_Y := \int_Y \rho(\bm{y}) d\bm{y}$ is the total mass of $Y$, and $\rho(\bm{y})$ is the mass density function over $Y$. When the polytope $Q_m$ is convex, the partitions are also convex. This in turn implies that the centroid of each partition is contained inside the polytope. For a fixed number of partitions $n_m$, a convex polytope $Q_m$ can have more than one centroidal Voronoi partition. While implementing this method for kernel center selection, the term $n_m$ corresponds to the number of centers. The subscript $m$ corresponds to the sampling subset $\Xi_m$. The number of kernel centers depends on the samples collected in this method. 
%
%
\subsection{Lloyd's algorithm}
Lloyd's algorithm is used to construct the centroidal Voronoi tessellations for a given convex polytope $Q_m$ and a fixed number of partitions $n_m$. It involves the following steps,
\begin{enumerate}[label=(\roman*)]
\item Choose an initial set of points $P_m$.
\item Calculate the Voronoi partitions $\mathcal{V}(P_m)$ for the $n_m$ points.
\item Calculate the set of centroids $\{ C_{P_{m,1}},\ldots,C_{P_{m,n_m}} \}$ of the Voronoi partitions.
\item Set $P_m = \{ C_{P_{m,1}},\ldots,C_{P_{m,n_m}} \}$ and go back to the second step.
\end{enumerate}
The above set of steps are evaluated until convergence of centroids is achieved. The convergence of the algorithm for the convex case is proved in \cite{Cortes2004}. 
%
\subsection{Implementation}

\begin{figure}[htb!]
\centering
\includegraphics[scale = 0.5]{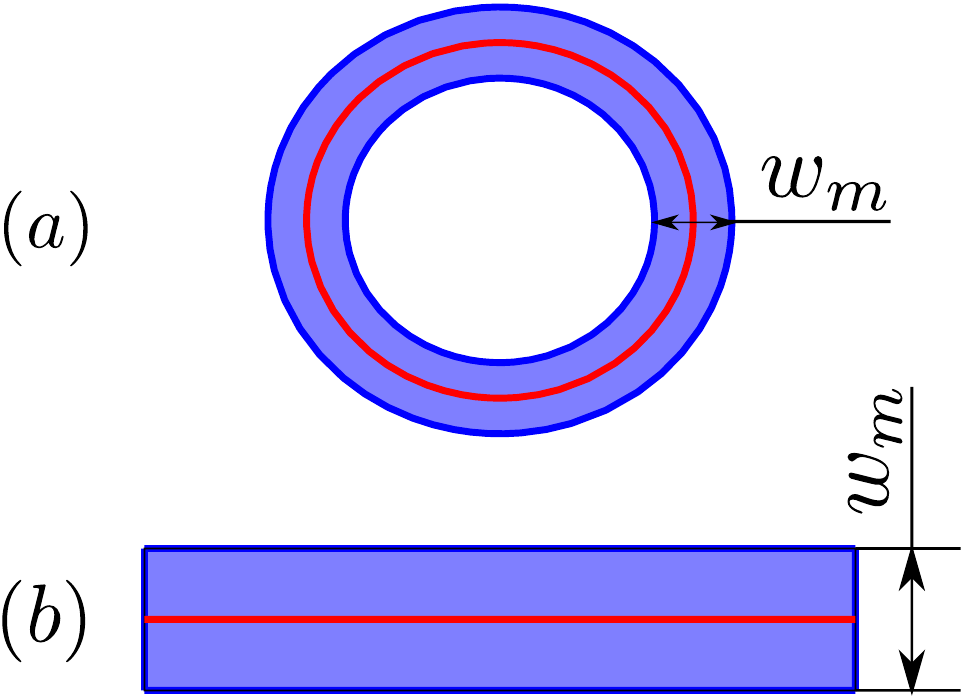}
\caption{Examples of region $Q_m$ constructed around the $\Xi_m \subseteq \omega^+(\bm{x}_0)$. The red curves are formed by connecting the samples $\Xi_m$. The blue region represents the region $Q_m$.}
\label{fig_LloydsRegion}
\end{figure}

The idea behind this approach is that we have a finite sampling $\Xi_m$ of the positive limit set $\omega^+(\bm{x}_0)$. We use this finite sampling $\Xi_m$ to construct a region $Q_m$ that encloses the positive limit set $\omega^+(\bm{x}_0)$. We then calculate the centroidal Voronoi partitions of the polygon and choose the kernel centers as the centroids of the partitions. In our implementation, we assume the mass density function as $\rho(\bm{q}) = 1$ for all $\bm{q} \in Q_m$ and $\rho(\bm{q}) = 0$ elsewhere. In the following discussion, we formalize this implementation.

Examples of the region $Q_m$ for two different positive limit sets is shown in Figure \ref{fig_LloydsRegion}. In the case (b) where the positive limit set $\omega^+(\bm{x}_0)$ is straight line, the region $Q_m$ is nothing but the rectangle enclosing the set. For the case (a) where the positive limit set $\omega^+(\bm{x}_0)$ is a closed curve that is symmetric about the origin in the figure, the region $Q_m$ is first formed by the joining the samples of the positive limit set to form a closed curve. The closed curve is then scaled to form a larger and smaller closed curves. We choose $Q_m$ to be the region enclosed by the larger and smaller closed curves. As evident from Figure \ref{fig_LloydsRegion}, the region $Q_m$ is not always convex. Thus, the theory in the previous subsection is not strictly applicable. Let $Q_m'$ be the convex hull of the polytope $Q_m$. We know that the Lloyd's algorithm converges for the convex case. \cite{Cortes2004} The mass density function is still equal to $1$ on $Q_m$ and $0$ elsewhere. Suppose we choose $n_m$ points in $Q_m'$ and run the Lloyd's algorithm. As a result, we get a set of centroids $P_m'$ that generate the centroidal Voronoi partition $\mathcal{V}(P_m')$. Now we define the collection $\mathcal{V}(P_m) := \{ P_{m,1}'\cap Q_m, \ldots, P_{m,n_m}' \cap Q_m \}$. It is easy to see that $\mathcal{V}(P_m)$ is a centroidal Voronoi partition of the region $Q_m$ generated by the centroids $P_m = P_m'$.

Thus, the Lloyd's algorithm indeed converges for the case in question. However, the polytopes in $\mathcal{V}(P_m)$ are not necessarily convex. And hence, the centroid $p_{m,i} \in P_m$ need not be contained in the polytope $P_{m,i}'\cap Q_m$ for $i = 1,\ldots,n_m$. The centers need not even be contained in the region $Q_m$. This is certainly not desirable when implementing Lloyd's algorithm and CVT for problems like sensor location or multirobot coordination. \cite{Breitenmoser2010} However, the goal of our problem is to choose kernel centers that are close to the positive limit set. In the following analysis, we show that with sufficient number of samples and careful selection of the the region $Q_m$, we can often choose centers close to the positive limit set.

\subsection{Convergence for Restricted Cases}
We restrict the following analysis to positive limit sets contained in $\mathbb{R}^2$ that are homeomorphic to a line or a circle. In other words, the positive limit set is an open or closed curve. With careful selection of $Q_m$, it is possible to show that we can choose kernel centers that approximate the positive limit set. The region $Q_m$ is constructed such that the following conditions holds.
\begin{condition}
\label{cond_lloyds}
Associated with each $\Xi_m$ is a region $Q_m$ such that 
\begin{enumerate}
    \item the maximum width $w_m$ of the region satisfies $w_m < r_m$, where $0< r_m < r_{m-1}$ for all $m \in \mathbb{N}$,
    \item the region $Q_m$ is nested in $Q_{m-1}$ for all $m \in \mathbb{N}$,
    \item the sequence $\{r_m\}_{n=1}^\infty$ converges to $0$,
    \item for each $r_m$, there is an integer $n_m$ such that the polytope $P_{m,j} \subseteq B_{c r_m}(C_{P_{m,j}})$ for all $j = 1,\ldots,n_m $. Here, the term $B_{c r_m}(C_{P_{m,j}})$ is the closed ball of radius $c r_m$ centered at the centroid $C_{P_{m,j}}$ that generates the polytope $P_{m,j}$ with $c$ a fixed positive constant.
\end{enumerate}
\end{condition}

We can think of the maximum width $w_m$ of the region $Q_m$ given in Figure \ref{fig_LloydsRegion} (a) as the Hausdorff distance between the inner and outer boundaries of the region $Q_m$. In the case of the region given in Figure \ref{fig_LloydsRegion} (b), the maximum width $w_m$ corresponds to the Hausdorff distance between the two boundaries of the region $Q_m$ that are parallel to the positive limit set.

\begin{theorem}
Suppose Condition \ref{cond_lloyds} holds. Then $d(\omega^+(\bm{x}_0), P_m) \to 0$ as $m \to \infty$, where $d(\cdot,\cdot)$ is the Hausdorff distance, $\omega^+(\bm{x}_0)$ is the positive limit set and $P_m = \{ C_{P_{m,1}},\ldots,C_{P_{m,n_m}} \}$ is the set of centroids that generate the CVT $\mathcal{V}(P_m)$.
\end{theorem}

\begin{proof}
We fist note that the centroid of each polytope is contained in $B_{c r_m}(C_{P_{m,j}})$ since the ball is convex. Since the maximum width of the region $w_m$ satisfies $w_m < r_m$, it is clear that $d(\omega^+(\bm{x}_0),Q_m) < r_m$. On the other hand, since the ball $B_{c r_m}(C_{P_{m,j}})$ contains the polytope $P_{m,j}$, we have $d(P_{m,j},\{C_{P_{m,j}}\}) < cr_m$ for any $j = 1,\ldots,n_m$. Note that the bound $cr_m$ on $d(P_{m,j},\{C_{P_{m,j}}\})$ is uniform. Also, recall that $Q_m = \cup_{j=1}^{n_m} P_{m,j}$, and $P_m = \cup_{j=1}^{n_m} \{ C_{P_{m,j}} \}$. Thus, we have $d(Q_m,P_m) < cr_m$. Using triangle inequality, we get $d(\omega^+(\bm{x}_0), P_m) < (1 + c)r_m$. Since $ r_m \to 0$ as $m \to \infty$, we conclude that the centroids approach the positive limit set as $m \to \infty$.
\end{proof}

The assumptions in the above theorem are very strong because of Condition \ref{cond_lloyds}. It is possible to relax some of the assumptions by considering the geometric properties of the partitions. But, from a practical standpoint, the maximum number of samples of the positive limit set is limited by the measurement equipment. This theorem provides a framework for an implementation that agrees with intuition - if new samples of the positive limit set are measured, choose $Q_m$ such that $r_m$ is reduced and number of kernel centers $n_m$ are increased. For a given $r_m$, the number of kernel centers cannot be indefinitely increased. Consider the example in Figure \ref{fig_LloydsNumErr}. Due to numerical errors, the Lloyd's algorithm converges to a CVT in which the kernel centers do not lie on the positive limit set when $n_m$ is large. On the other hand, the term $r_m$ cannot be decreased indefinitely, since the region $Q_m$, built based on finite number of samples, may no longer contain the positive limit set. Thus, the number of samples collected restrict the effectiveness of this method.

To avoids CVTs that are similar to the one given in Figure \ref{fig_LloydsNumErr} (b), we introduce the following condition. Let $\bar{Q}$ represent the outer rectangle that is contained in $\mathbb{R}^2$ in Figure \ref{fig_LloydsNumErr} and let $\bar{\mathcal{V}}_l$ represent the CVT made up of $l$ horizontally stacked identical rectangles. Figure \ref{fig_LloydsNumErr} (a) depicts the CVT $\bar{\mathcal{V}}_5$ of $\bar{Q}$. The following condition inherently ensures that the kernel centers are evenly distributed in or near the positive limit set.

\begin{condition}
\label{cond_CVTtype}
Let $l = 1,\ldots,n_m$. For any possible $l$, consider an arbitrary collection of $l$ polytopes $P_{m,i_1}, \ldots, P_{m,i_l}$ in the partition $\mathcal{V}(P_m)$ such that each polytope is adjacent to at least one other polytope in the collection. The union of edges $\mathbb{E}(\mathcal{V}(P_m),P_{m,i_1} \cup \ldots \cup P_{m,i_l})$ is homeomorphic to the union of edges $\mathbb{E}(\bar{\mathcal{V}}_l,\bar{Q})$ of the CVT $\bar{\mathcal{V}}_l$.
\end{condition}

\begin{figure}[htb!]
\centering
\includegraphics[scale = 0.5]{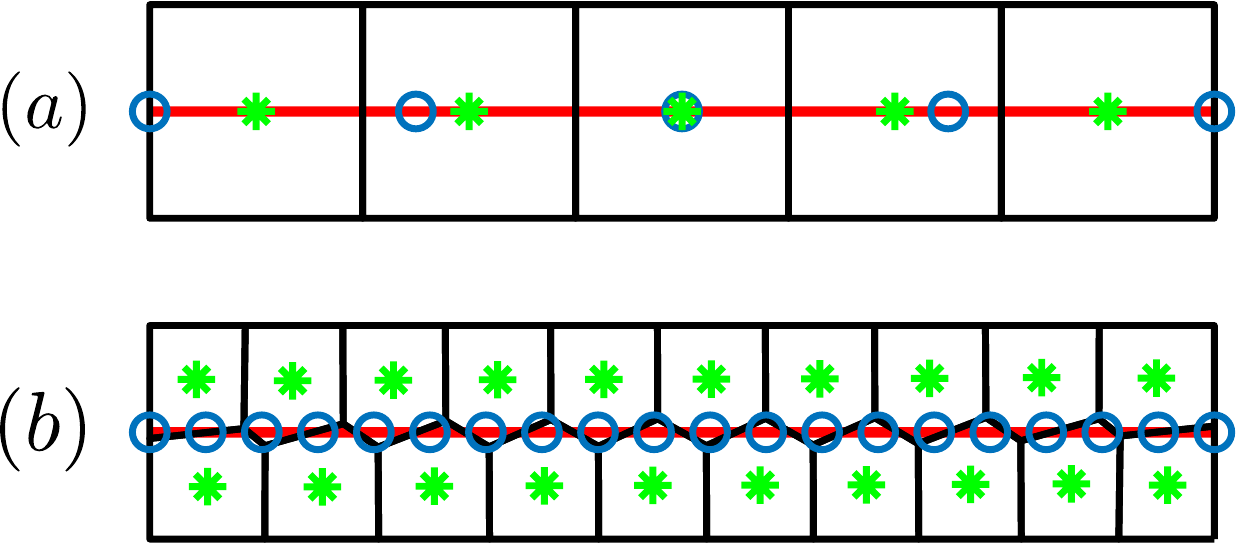}
\caption{Increasing the number of kernel centers leads to completely different types of CVT while using the same Lloyds algorithm. The markers $\mathrm{o}$ and $*$ represent the initial positions and final converged positions of the kernel centers, respectively. The red line represents the limit set.}
\label{fig_LloydsNumErr}
\end{figure}

\begin{algorithm}
\KwIn{$\Xi_m$, $n_m$}
\KwOut{$P_m$}
Choose the constant $r_m$. Construct region $Q_m$ such that the positive limit set $\omega^+(\bm{x}_0)$ is contained in $Q_m$.
\\
Choose $n_m$ separate points in the convex hull of $Q_m$.
\label{algStp_seppts}
\\
Run the Lloyd's algorithm using the points chosen in Step \ref{algStp_seppts} as the initial points.
{
\begin{enumerate}[label=(\roman*)]
\item Calculate the Voronoi partitions $\mathcal{V}(P_m)$ for the $n_m$ points.
\item Calculate the centroids $C_{P_{m,1}},\ldots,C_{P_{m,n_m}}$ of the Voronoi partitions $\mathcal{V}(P_m)$.
\item Set $P_m = \{ C_{P_{m,1}},\ldots,C_{P_{m,n_m}} \}$ and go back to the Step 3 (i).
\end{enumerate}
}
The above steps are repeated until convergence is achieved. \label{algStp_Lloyds}
\\
If the CVT from Step \ref{algStp_Lloyds} does not satisfy Condition \ref{cond_CVTtype}, choose a constant $s_m$ such that $s_m < r_m$. Set $r_m = s_m$ and go back to Step \ref{algStp_seppts}.
\linebreak
\noindent If the CVT satisfies Condition \ref{cond_CVTtype}, choose the set of centroids of the CVT $P_m$ as the kernel centers for the adaptive estimator.
\caption{CVT based kernel center selection}
\label{alg_CVT}
\end{algorithm}

Algorithm \ref{alg_CVT} shows the steps involved in implementing this method. Step \ref{algStp_Lloyds} in the algorithm can be implemented using commercially available tools like MATLAB, which makes the algorithm extremely straightforward for implementation. The inputs to the algorithm are the samples $\Xi_m$ and the number of kernel centers $n_m$. We iteratively choose $r_m$ in the algorithm until Condition \ref{cond_CVTtype} is satisfied. The output of the algorithm is the set of kernel centers, which can be implemented in the adaptive estimator algorithm.

%
\section{Method 2: Based on Kohonen Self-Organizing Maps}
\label{sec_Koho}
The second approach presented in this paper is based on Kohonen self-organizing maps (SOMs), which were first introduced by Teuvo Kohonen. \cite{Kohonen2012} Self-organizing maps are typically used for applications like clustering data, dimensionality reduction, pattern recognition, and visualization. Thus, given a set of samples in the input space, these maps can be used to produce a collection of neurons on a low-dimensional manifold that represents the samples' distribution. In our problem, the input space is the state-space, and the samples are the state measurements. The neurons on the low-dimensional manifold are the kernels centers. The position of the kernel centers in the state-space are represented by the weight vectors that the SOM algorithm generates.

One of the critical features of self-organizing maps is that the underlying topology between the input space (the original dataset) and the output space is maintained. Intuitively, points that are close in the original dataset are mapped to neurons that are close to each other (in some predefined metric). For our problem, we want the kernel centers to be evenly spaced in the state-space in addition to being close to the measurement samples. To ensure this, we choose the initial set of kernel centers on a manifold that is homeomorphic to the positive limit set. This requires knowledge of the topology of the positive limit set. Before going over the details, let us take a look at the theory of Kohonen self-organizing maps. 

Suppose we have the set of samples $\Xi_m = \{ \bm{\xi}_{m,1}, \ldots, \bm{\xi}_{m,q_m} \}$. In the context of this paper, the set $\Xi_m$ is the set of samples of the positive limit set $\omega^+(\bm{x}_0)$. Let $n_m$ represent the number of kernel centers $\bm{p}_{m,1},\ldots,\bm{p}_{m,n_m}$ we want to choose. We associate the $i^{th}$ kernel center with a weight vector $\bm{p}_{m,i}(t) \in \mathbb{R}^d$ for $i = 1,\ldots,n_m$. Note that the weight vectors depend on time and at any given instant in time $t$, the weight vector is an element of $\mathbb{R}^d$. The neighborhood function $\mathcal{N}_j$ defines neighbors of the center $j$. The choice of the neighborhood function depends on the topology we want to define on the kernel centers. The neurons (or the kernel centers) are often chosen in the form of a linear grid or a $2D$ grid, and the neighbors in such grids are naturally defined. The Kohonen self-organizing map's implementation involves the following steps. We first randomly choose a sample $\bm{\xi}_{m,k}$ from the sample set $\Xi_m$, where $k \in \{ 1, \ldots, q_m \}$. We then determine the winning neuron - the kernel center that is closest to the sample $\bm{\xi}_{m,k}$. The winning neuron $i$ at a given instant $t$ is the one which satisfies the condition
\begin{align}
\label{eq_SOMwinNode}
d( \bm{\xi}_{m,k}, \bm{p}_{m,i}(t) ) \leq d( \bm{\xi}_{m,k}, \bm{p}_{m,j}(t) )
\end{align}
for $j = 1, \ldots, n_m$. We now update the weight vectors using the evolution equation
\begin{align}
\label{eq_SOMUpdate}
\frac{d \bm{p}_{m,j}(t)}{dt} = \beta_j(t) \mathcal{N}_j(t,i) \left( \bm{\xi}_{m,k} - \bm{p}_{m,j}(t) \right)
\end{align}
for $j = 1, \ldots, n_m$. In the above equation, $0 \leq \beta_j(t) < 1$ defines the rate of convergence of the center $j$. The neighborhood function determines which neighbors of the node $i$ get updated. For convergence, we require that $\beta_j(t) \to 0$ and $\mathcal{N}_j(t,i) \to 0$ as $t \to \infty$, for any $i, j \in \{ 1, \ldots, n_m \}$. 
While implementing this algorithm, we can observe the SOM goes through a \emph{topological ordering phase} during which the grid of neurons try to match the patterns if the sample in the input space before convergence.

\textit{Note:} The self-organizing map algorithm is easy to implement. However, many theoretical aspects of these maps, like convergence, remain unanswered for the general case. Researchers have studied and proved the theory for the 1D linear array case, when the nodes are arranged on a line. A review of some of the theoretical results are in \cite{Cottrell1998}.

%
\subsection{Implementation}
\label{ssec_KohoImp}
To implement Kohonen self-organizing maps for kernel center selection, we modify the above-discussed algorithm. In some dynamical systems, the trajectory approaches the positive limit set but is never contained in the set. In such cases, we only have measurements of the states and not the samples of positive limit set. Furthermore, arbitrary selection of state-samples might result in picking points away from the positive limit set. This in turn affects the convergence of the kernel centers to points inside the positive limit set. Hence, as opposed to choosing random samples $\bm{\xi}_{m,j}$ from the set $\Xi_m$, we use the state measurement $\bm{x}(t)$ at a given time instant to determine the winning node. We replace the term $\bm{\xi}_{m,j}$ with $\bm{x}(t)$ in Equations \ref{eq_SOMwinNode} and \ref{eq_SOMUpdate}. This change enables us to implement this method for a more general class of systems in real-time. 

A Kohonen self-organizing map algorithm gives a low-dimensional representation of all samples (which include the ones that are outside the limit set). On the other hand, the objective of our problem is to choose kernel centers on the positive limit set such that they are spaced as uniformly as possible. To ensure this, we choose the topology of the output space to match that of the positive limit set. In other words, we choose the initial kernel centers and the neighborhood function such that the topology is homeomorphic to the positive limit set. For example, if the positive limit set is a closed curve in $\mathbb{R}^2$, the initial weight vectors can be points on the unit circle, and the neighborhood function can be defined as 
\begin{align}
    \mathcal{N}_j(t,i) = 
    \left \{
    \begin{array}{cc}
        1 & \text{if } j \in \mathcal{T}, \\
        0 & \text{if } j \notin \mathcal{T},
    \end{array}
    \right.
    \label{eq_KohoNbhd}
\end{align}
where the set $\mathcal{T}$ is defined as $\mathcal{T} = \{i-1,i,i+1\}$ for $i \neq 1, n_m$. For $i = 1$ and $i = n_m$, we choose $\mathcal{T} = \{n_m,1,2\}$ and $\mathcal{T} = \{ n_m - 1, n_m, 1 \}$, respectively.

On top of the above modifications, we enforce the condition that, when we have samples of the positive limit set, the number of kernel centers or neurons $n_m$ should be strictly less than $q_m$, the number of samples in the set $\Xi_n$. When $n_m$ is equal to $q_m$, the kernel centers can converge to the samples. In the case where the positive limit set is a closed curve, this can be interpreted as a solution to the traveling salesman problem. \cite{Brocki2007} To avoid convergence to the samples, we impose the above dimensionality reduction condition. 

Algorithm \ref{alg_Koho} shows the steps involved in implementing this method. We present the algorithm for the case where the positive limit set is a closed curve. However, the algorithm can be extended easily for other types of positive limit sets. The neighborhood function for this case, defined by Equation \ref{eq_KohoNbhd}, is inherently accounted in the algorithm.


\begin{algorithm}
\KwIn{$\bm{x}(t)$, $q_m$}
\KwOut{$\{ \bm{p}_{m,1}(T), \ldots, \bm{p}_{m,n_m}(T)  \}$}
Choose the number of kernel centers $n_m$ such that $n_m < q_m$. If $p_m = 0$, choose a positive integer for $n_m$.
\\
Choose $\beta_j$ such that $0 \leq \beta_j(t) < 1$ for $t \in [0,\infty)$ and $\beta_j(t) \to 0$ as $t \to \infty$ for all $j = 1, \ldots, n_m$.
\\
Initialize the weight vectors $\bm{p}_{m,j}$ as the points on a circle contained inside the closed curve.
\\
Implement the Kohonen SOM algorithm for $t \in [0,T]$ for some $T>0$.
{
\begin{enumerate}[label=(\roman*)]
    \item At time $t$, determine the winning neuron $i$ that satisfies the condition
    $$
    d(\bm{x}(t) - \bm{p}_{m,i}(t)) \leq d(\bm{x}(t) - \bm{p}_{m,j}(t))
    $$
    for $j = \{ 1, \ldots, n_m \}$. 
    \item
    Define the set $\mathcal{T}$ as $\mathcal{T} = \{i-1,i,i+1\}$ for $i \neq 1, n_m$. For $i = 1$ and $i = n_m$, choose $\mathcal{T} = \{n_m,1,2\}$ and $\mathcal{T} = \{ n_m - 1, n_m, 1 \}$, respectively.
    \item Update the weight vectors based on
    \begin{align*}
    \frac{d \bm{p}_{m,j}(t)}{dt} = 
    \left \{
    \begin{array}{lc}
    \beta_j(t) \left( \bm{x}(t) - \bm{p}_{m,j}(t) \right) & \text{if } j \in \mathcal{T}
    \\
    0 & \text{if } j \notin \mathcal{T}
    \end{array}
    \right.
    \end{align*}
    for $j = 1, \ldots, n_m$. This update happens until next state measurement. Go back to Step 4 (i) after the update.
\end{enumerate}
}
\caption{Kohonen SOM based Kernel Center Selection - Closed Curve Case}
\label{alg_Koho}
\end{algorithm}

Recall that in the case of CVT based method presented in the previous section, the samples are contained in the positive limit set, which meant the trajectory was contained in the positive limit set or converged to the set in finite time. Since we use the state measurement for the Kohonen SOM based approach, we can relax some of the requirements of the CVT based method. It is sufficient for the trajectory to converge to the positive limit set as $t \to \infty$. However, it is important to choose $\beta_j(t)$ such that the state trajectory converges to the positive limit set faster than the rate at which $\beta_j(t) \to 0$. If this is violated, the kernel centers will not converge to the positive limit set.

Note, in the Lloyd's algorithm, the distance between any two kernel centers is inherently ensured to remain uniform by the algorithm. This can be attributed to the way partitions are defined and the selection of the mass density function. On the other hand, the distribution of the converged kernel centers from the Kohonen SOM based algorithm depends on the distribution of the sampled measurements. If the state measurements are concentrated on a particular neighborhood of the positive limit set, implementing Algorithm \ref{alg_Koho} will result in the kernel centers being concentrated in or near the neighborhood.

%
\section{Numerical Illustration of Center Selection Methods}
\label{sec_examples}
We illustrate the effectiveness of the two approaches explained above for two examples in this section. The first example is the undamped piezoelectric oscillator example considered in Section \ref{ssec_sampex}. The positive limit set in this case is almost symmetric about the axis after scaling of the states. The second example is a nonlinear oscillator which has a nonsymmetric positive limit set. We implement the above discussed methods for both cases and use the resulting kernel centers in the adaptive estimators. We use MATLAB \mcode{lloydsAlgorithm} function, developed by Aaron T. Becker's Robot Swarm Lab, for implementing Step \ref{algStp_Lloyds} of Algorithm \ref{alg_CVT}. 
The function expects the boundary of a polygon as input and hence we approximate the region $Q_m$ using a polygon as shown in Figures \ref{fig_piezoConv} and \ref{fig_nonOscEx}. In the adaptive estimator simulations, we use the Sobolev-Matern $3,2$ kernel given in Subsection \ref{ssec_sampex}.
%
\begin{figure*}
\centering
\begin{subfigure}{1\textwidth}
\centering
\includegraphics[scale = 0.6]{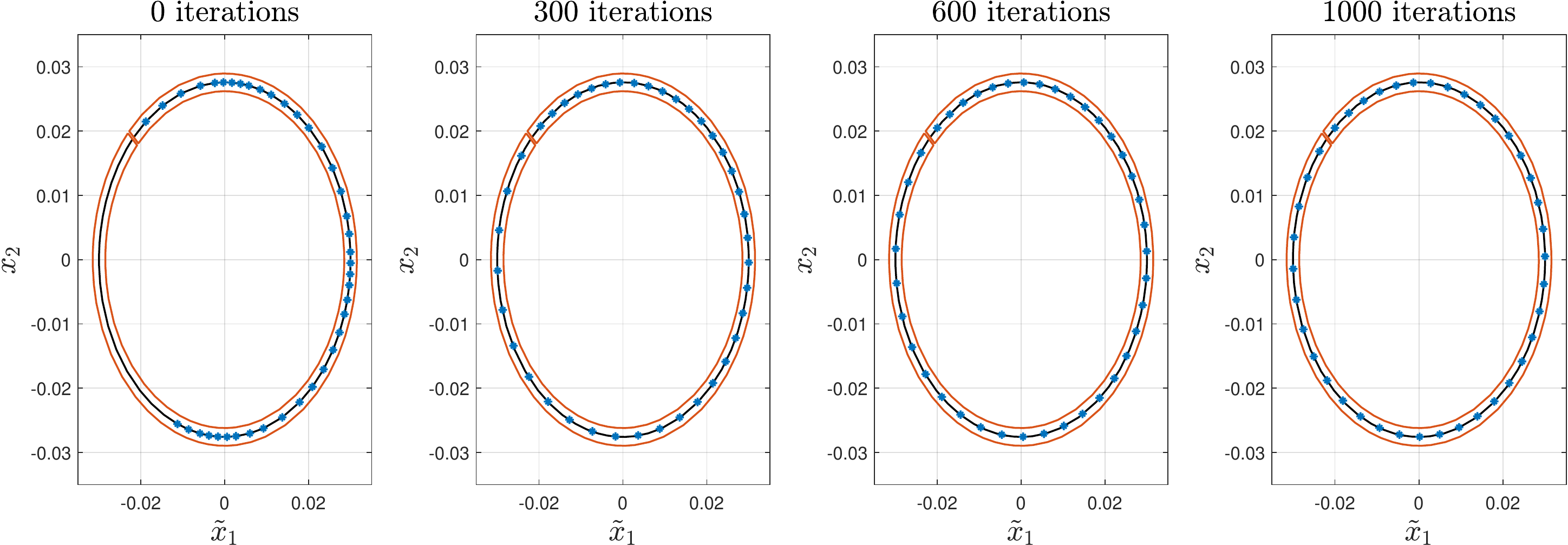}
\caption{Algorithm \ref{alg_CVT} output}
\label{sfig_piezoConvLloyds}
\end{subfigure}
\begin{subfigure}{1\textwidth}
\centering
\includegraphics[scale = 0.6]{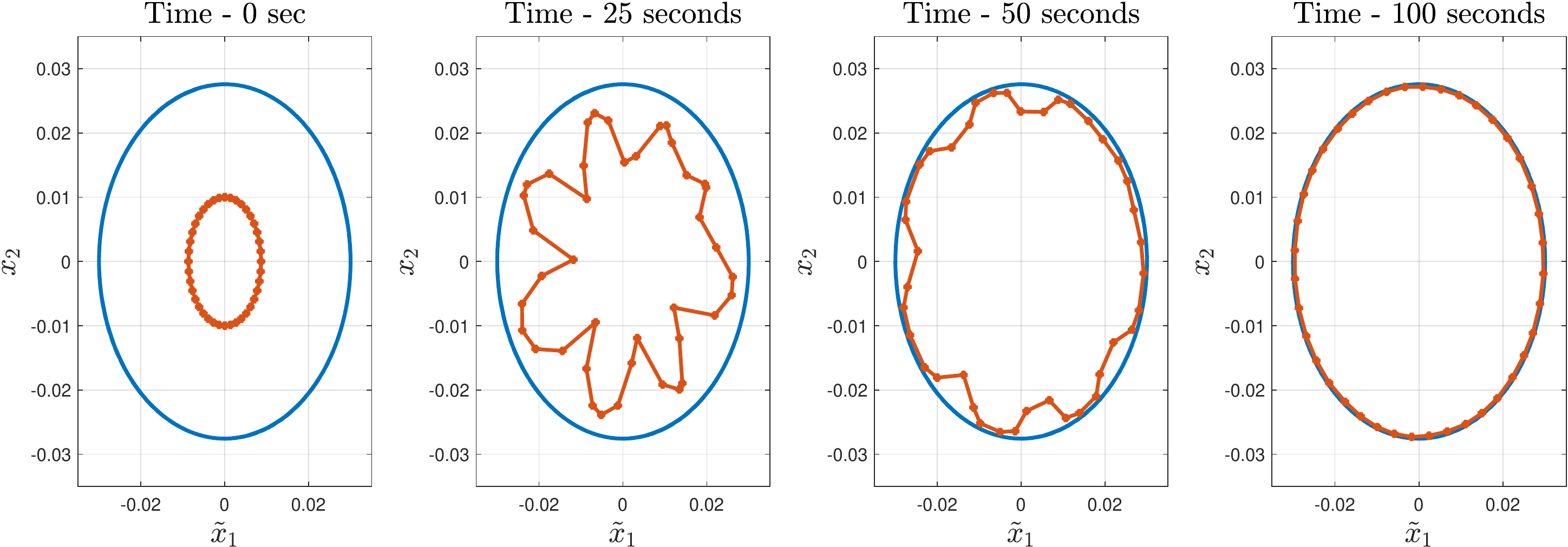}
\caption{Algorithm \ref{alg_Koho} output}
\label{sfig_piezoConvKoho}
\end{subfigure}
\caption{Algorithm outputs of Example \ref{ssec_NumExPiezo}. The marker $*$ and the blue line represent the kernel centers and the limit set, respectively.}
\label{fig_piezoConv}
\end{figure*}
\begin{figure}
\centering
\includegraphics[scale = 0.5]{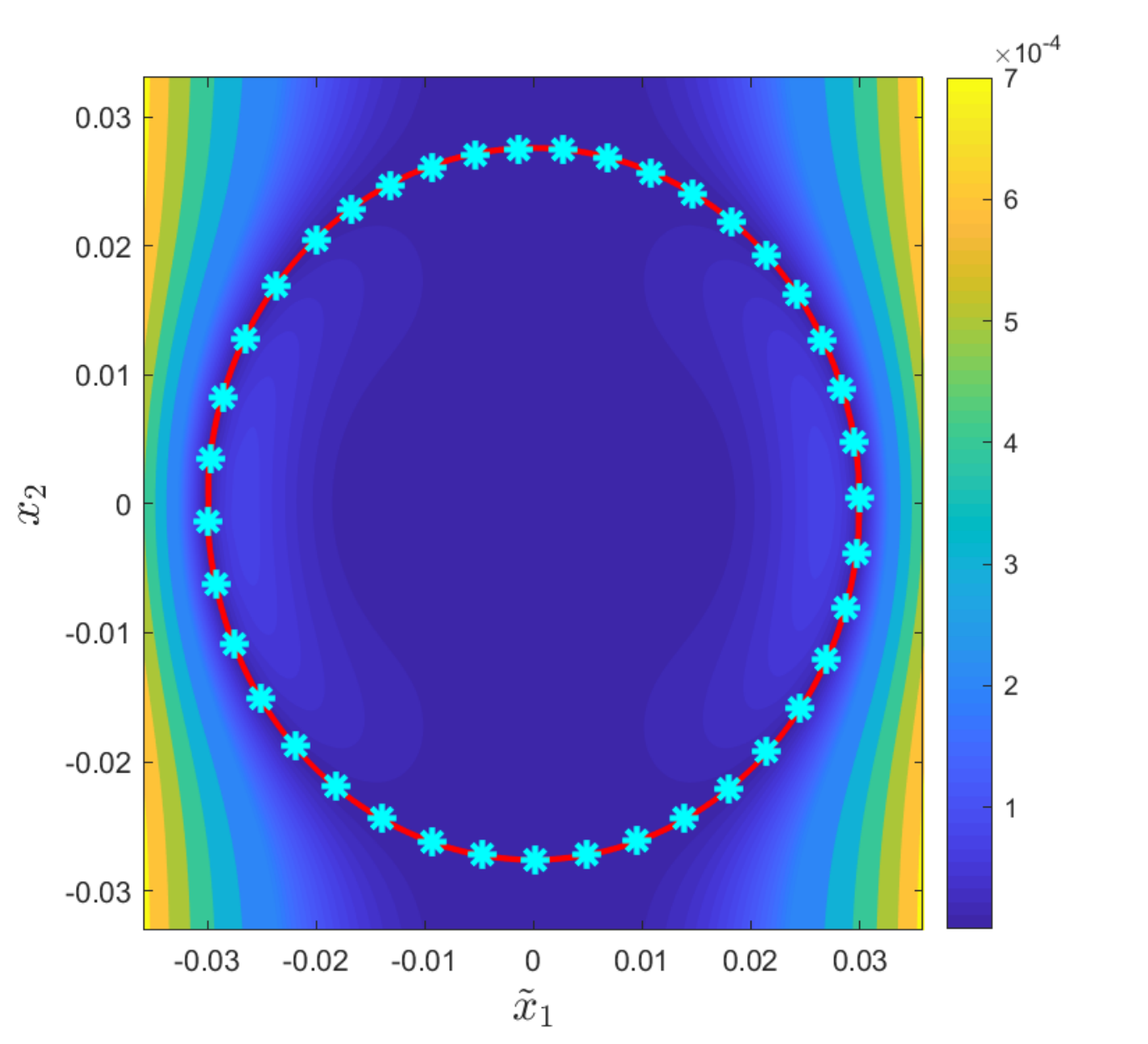}
\caption{Kernel centers for Example \ref{ssec_NumExPiezo} selected using Algorithm \ref{alg_CVT}  - Pointwise error $|f(\bm{x}) - \hat{f}_n(T,\bm{x})|$ obtained from adaptive estimator. The marker $*$ and the red line represent the kernel centers and the limit set, respectively.}
\label{fig_piezopwErrLloyds}
\end{figure}
\begin{figure}
\centering
\includegraphics[scale = 0.5]{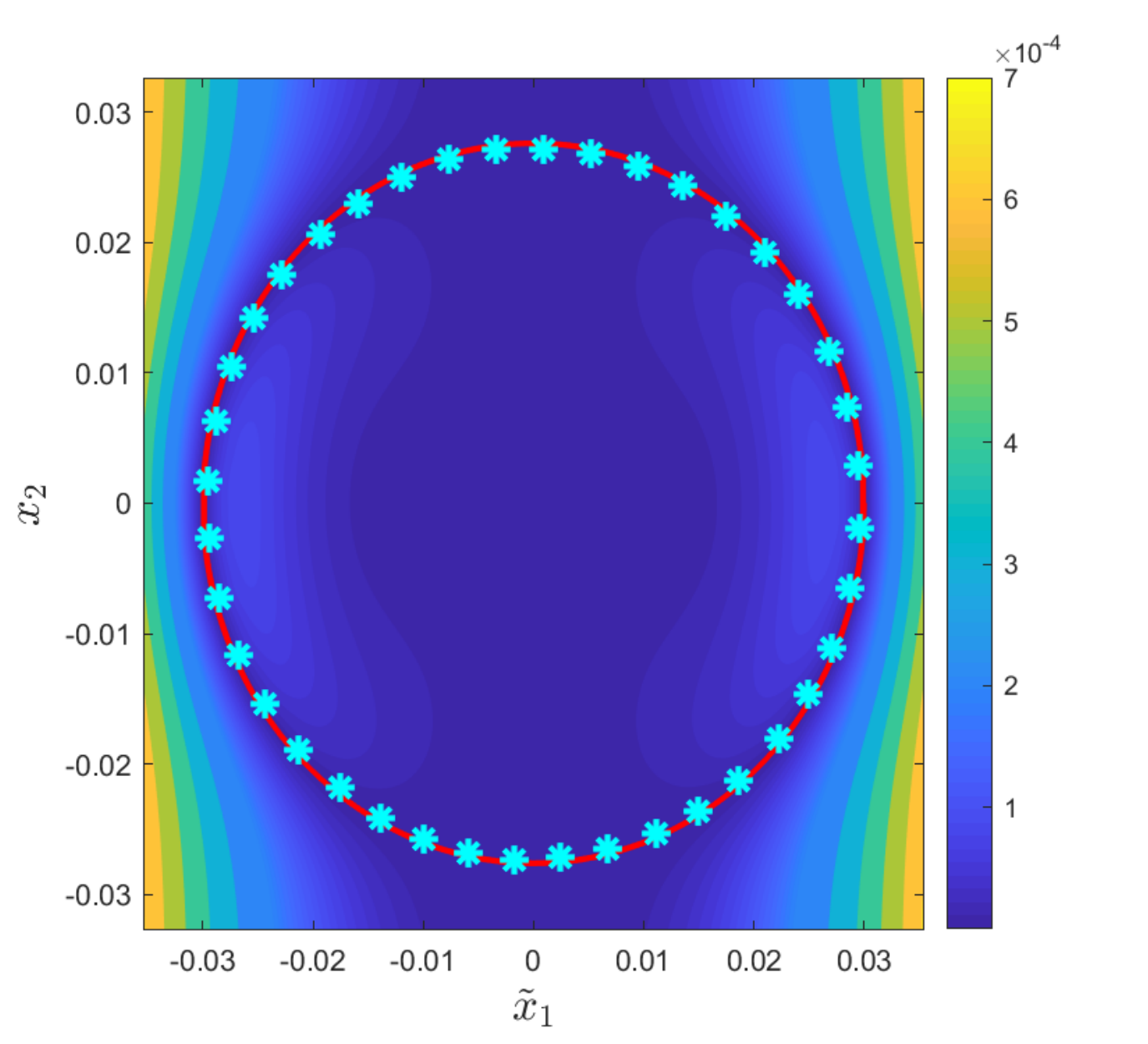}
\caption{Kernel centers for Example \ref{ssec_NumExPiezo} selected using Algorithm \ref{alg_Koho}  - Pointwise error $|f(\bm{x}) - \hat{f}_n(T,\bm{x})|$ obtained from adaptive estimator. The marker $*$ and the red line represent the kernel centers and the limit set, respectively.}
\label{fig_piezopwErrKoho}
\end{figure}
%
\subsection{Example 1: Nonlinear Piezoelectric Oscillator}
\label{ssec_NumExPiezo}
The first example we consider is the undamped nonlinear piezoelectric oscillator whose motion is governed by the Equation \ref{eq_piezomodel}. We use the same values for the structural parameters as the ones used in the example in Section \ref{ssec_sampex}. We set the scaling factor $S = 0.02$ and initialized the states at $\bm{x}_0 = \{ \tilde{x}_1(0),x_2(0) \}^T = \{ 0.03,0 \}^T$. Figure \ref{fig_piezoConv} shows how the kernel centers evolve while using Algorithms \ref{alg_CVT} and \ref{alg_Koho}. We set the number of kernel centers as $n_m = 40$ for both of the algorithms. For implementing Algorithm \ref{alg_CVT}, we first collect the set of samples $\Xi_m$ of the positive limit set $\omega^+(\bm{x}_0)$. By connecting the samples in $\Xi_m$ with straight lines, we form a closed curve which is represented by the blue line in Figure \ref{sfig_piezoConvLloyds}. We then scale the closed curve by a factor of $1.1$ and $0.9$, thus forming concentric larger and smaller closed curves. We chose the region between these two closed curves as $Q_m$. Dividing the region $Q_m$ as shown in Figure \ref{sfig_piezoConvLloyds} results in a polygon, thus enabling us to use the \mcode{lloydsAlgorithm} function in MATLAB. While implementing Algorithm \ref{alg_CVT}, we chose $\beta_j(t) = 0.99$ for $t \leq 1000$ s and $\beta_j(t) = 0$ for $t> 1000$ s for all $j$. As evident from Figure \ref{fig_piezoConv}, the CVT based approach and the Kohonen SOM based approach take $1000$ iterations and $100$ seconds, respectively to converge. It is clear that the kernel centers are more uniformly spaced than those picked arbitrarily in the example in Subsection \ref{ssec_sampex}. We subsequently use the converged kernel centers and simulate the adaptive estimator algorithm for $T = 300$ seconds. For the adaptive estimator, we set $l = 0.006$, $\Gamma = 0.001$ and initialized the parameters at $\alpha_i(t) = 0.0001$ for $i = 1\ldots,n_m$. Figures \ref{fig_piezopwErrLloyds} and \ref{fig_piezopwErrKoho} shows the pointwise error $|f(\bm{x}) - \hat{f}(T,\bm{x})|$ obtained after using the kernel centers from the CVT and Kohonen SOM based approach. As expected, both the plots show that the error is $\mathcal{O}(10^{-4})$ over the positive limit set.

%
\subsection{Example 2: Nonlinear Oscillator}
\label{ssec_NumExNonOsc}
For the second example, we consider a nonlinear oscillator whose motion is governed by the equation
\begin{align}
    \begin{Bmatrix}
    \dot{x}_1 \\ \dot{x}_2
    \end{Bmatrix}
    &=
    \underbrace{\begin{bmatrix}
    0 & 1 \\
    -1 & 0.5
    \end{bmatrix}}_{A}
    \begin{Bmatrix}
    x_1 \\ x_2
    \end{Bmatrix}
    +
    \underbrace{\begin{Bmatrix}
    0 \\ 1
    \end{Bmatrix}}_{B}
    \underbrace{ 
    \left(
    -x_1^2 x_2
    \right)
    }_{f(\bm{x}(t))}.
    \label{eq_nonOsc}
\end{align}
This system exhibits a more complex behavior than that in Example \ref{ssec_NumExPiezo}. Firstly, the state trajectory is not contained in the positive limit set $\omega^+(\bm{x}_0)$, which is depicted as the blue, solid line in Figure \ref{fig_nonOscEx}. Note that the positive limit set is not symmetric. Refer Example 9.2.2 in \cite{Hubbard2013} for a detailed analysis of the nonlinear behavior of the oscillator. Here, we are interested in estimating the nonlinear function $f(\bm{x}(t)) = -x_1^2 x_2$.
\begin{figure*}
\centering
\begin{subfigure}{1\textwidth}
\centering
\includegraphics[scale = 0.6]{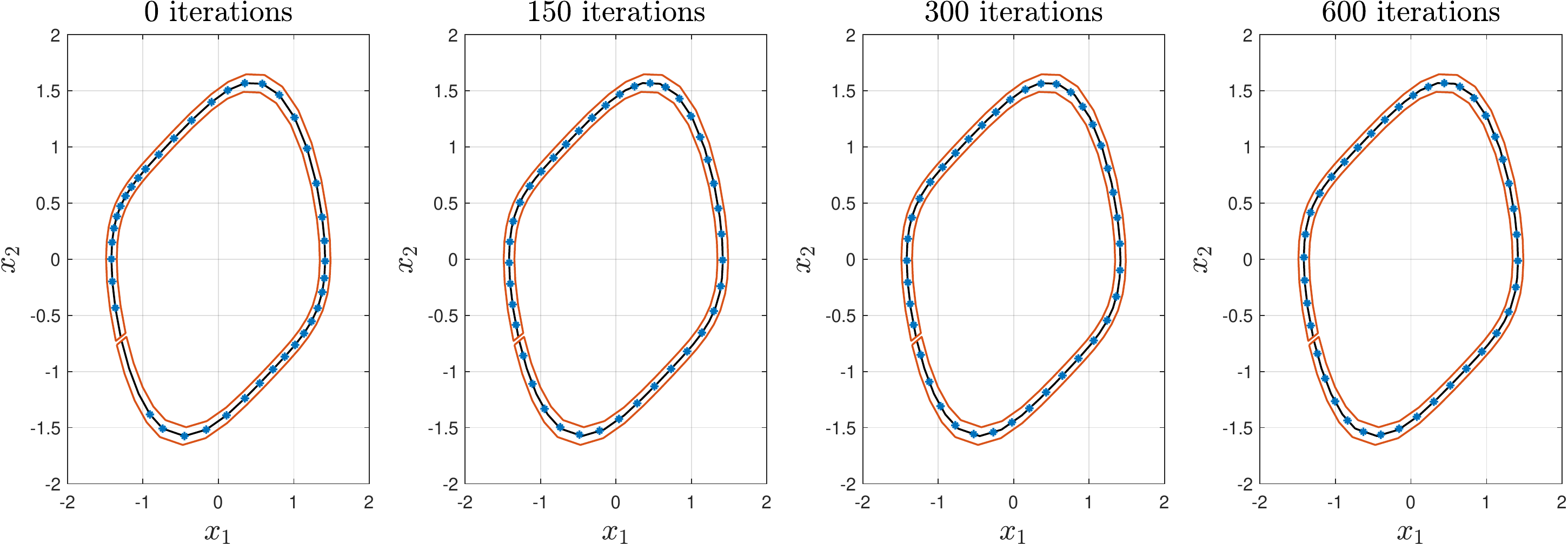}
\caption{Algorithm \ref{alg_CVT} output}
\label{sfig_NonOscConvLloyds}
\end{subfigure}
\begin{subfigure}{1\textwidth}
\centering
\includegraphics[scale = 0.6]{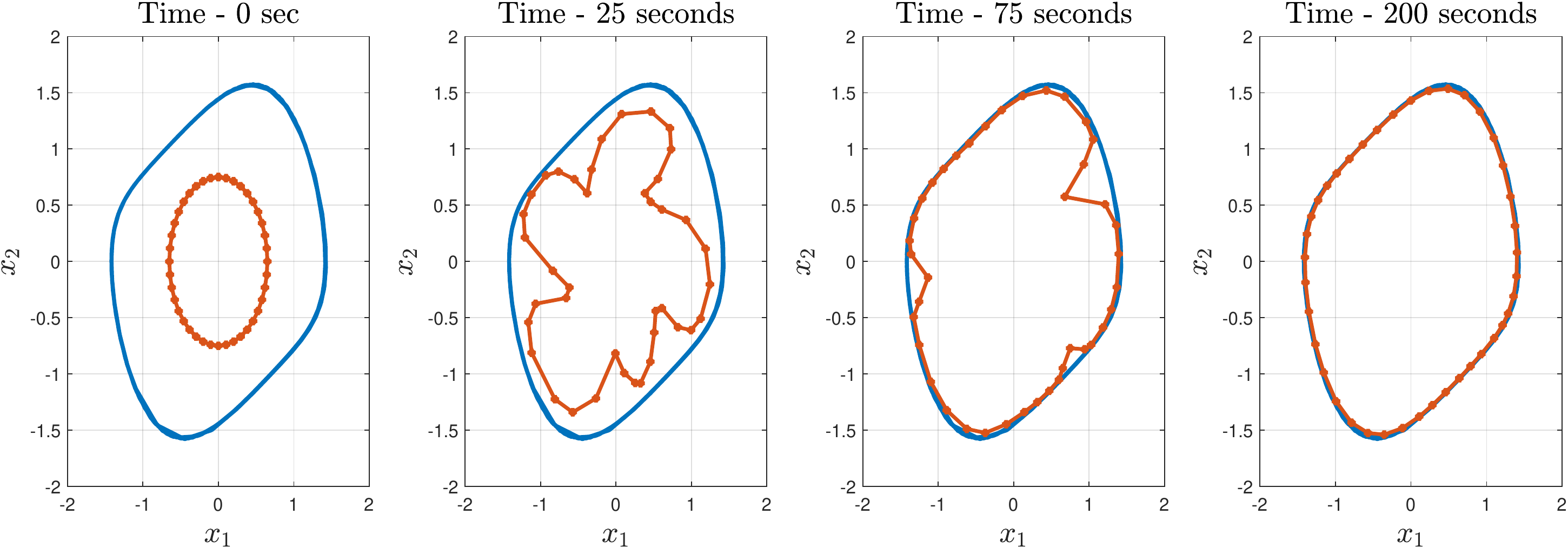}
\caption{Algorithm \ref{alg_Koho} output}
\label{sfig_NonOscConvKoho}
\end{subfigure}
\caption{Algorithm outputs of Example \ref{ssec_NumExNonOsc}. The marker $*$ and the blue line represent the kernel centers and the limit set, respectively.}
\label{fig_nonOscEx}
\end{figure*}
Figure \ref{fig_nonOscEx} shows the implementation of the CVT based and Kohonen SOM based kernel center selection methods for this problem. In both cases, we fixed number of kernel center as $n_m = 40$ and initialized the states at $\bm{x}_0 = \{ x_1(0), x_2(0) \}^T = \{0 , 2 \}^T$. The polygon in Figure \ref{sfig_NonOscConvLloyds} for the CVT based approach is built similar to the method used for Example \ref{ssec_NumExPiezo}. For the Kohonen SOM approach, we set $\beta_j(t) = 0.99$ for $t \leq 1000$ s and $\beta_j(t) = 0$ for $t> 1000$ s for all $j$. As evident from the figures, the CVT and Kohonen SOM methods take $600$ iterations and $200$ seconds, respectively for convergence of the kernel centers. It is clear that the kernel centers from the CVT based algorithm are more uniformly placed that the output of the Kohonen SOM algorithm. This can be attributed to the fact the state measurement samples are not uniformly distributed and to the fact that the CVT method makes strong assumptions about the structure of $Q_m$. Since the distribution of the state measurement affect the results of the Kohonen SOM based approach, the kernel centers are not uniform in this case. However, when the kernel centers from these algorithms are implemented in the adaptive estimator, we obtain convergence on the positive limit set. Figures \ref{fig_nonOscpwErrLloyds} and \ref{fig_nonOscpwErrKoho} shows the pointwise error $|f(\bm{x}) - \hat{f}(T,\bm{x})|$ after implementing the adaptive estimator for $T = 300$ seconds using the kernel centers from the CVT and Kohonen SOM based kernel center selection approach, respectively. We set $l = 0.5$, $\Gamma = 0.001$ and initialized the parameters at $\alpha_i(t) = 0.0001$ for $i = 1\ldots,n_m$. As in Example \ref{ssec_NumExPiezo}, the error is the smallest over the positive limit set.
\begin{figure}
\centering
\includegraphics[scale = 0.5]{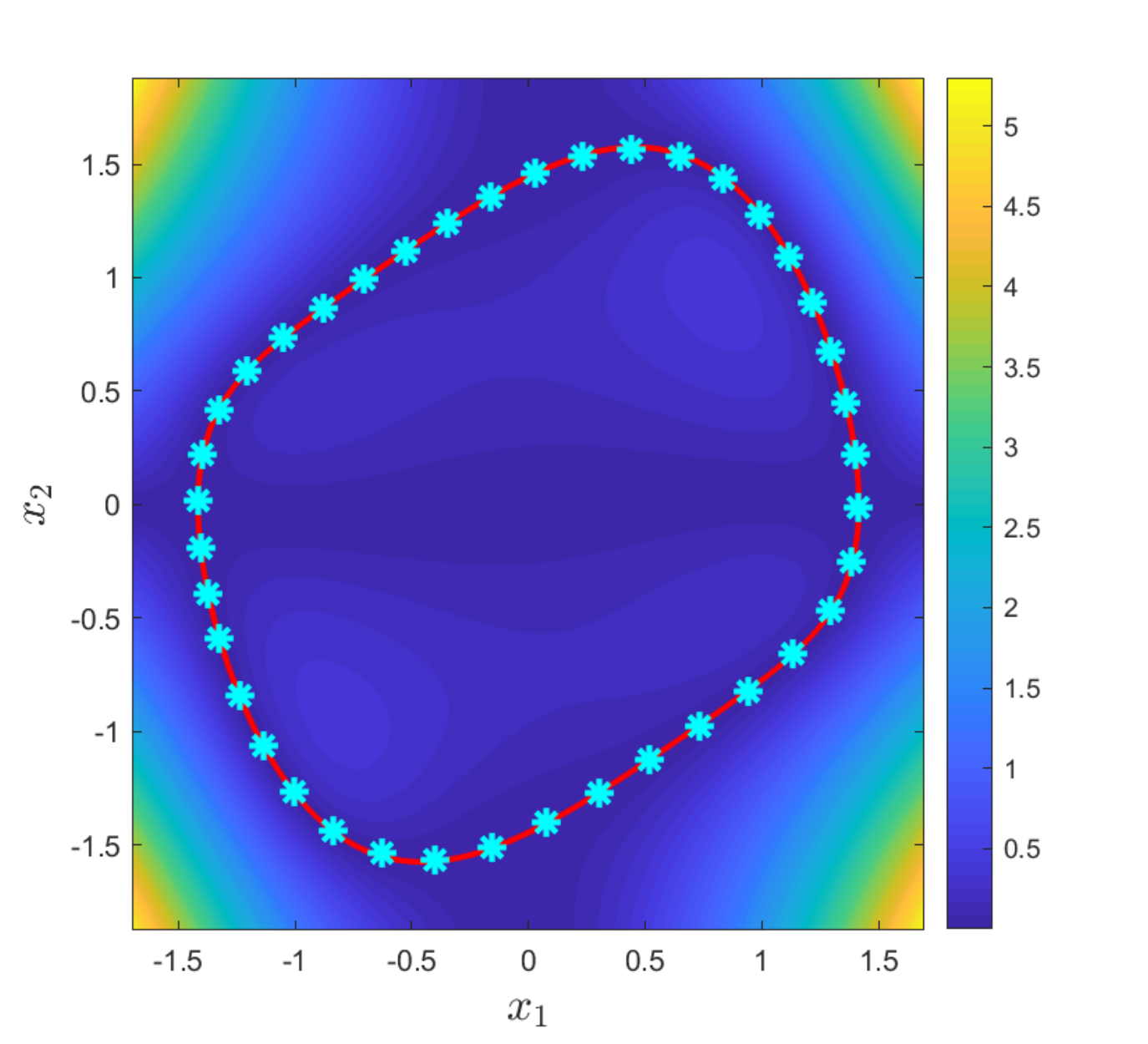}
\caption{Kernel centers for Example \ref{ssec_NumExNonOsc} selected using Algorithm \ref{alg_CVT}  - Pointwise error $|f(\bm{x}) - \hat{f}_n(T,\bm{x})|$ obtained from adaptive estimator. The marker $*$ and the red line represent the kernel centers and the limit set, respectively.}
\label{fig_nonOscpwErrLloyds}
\end{figure}
\begin{figure}
\centering
\includegraphics[scale = 0.5]{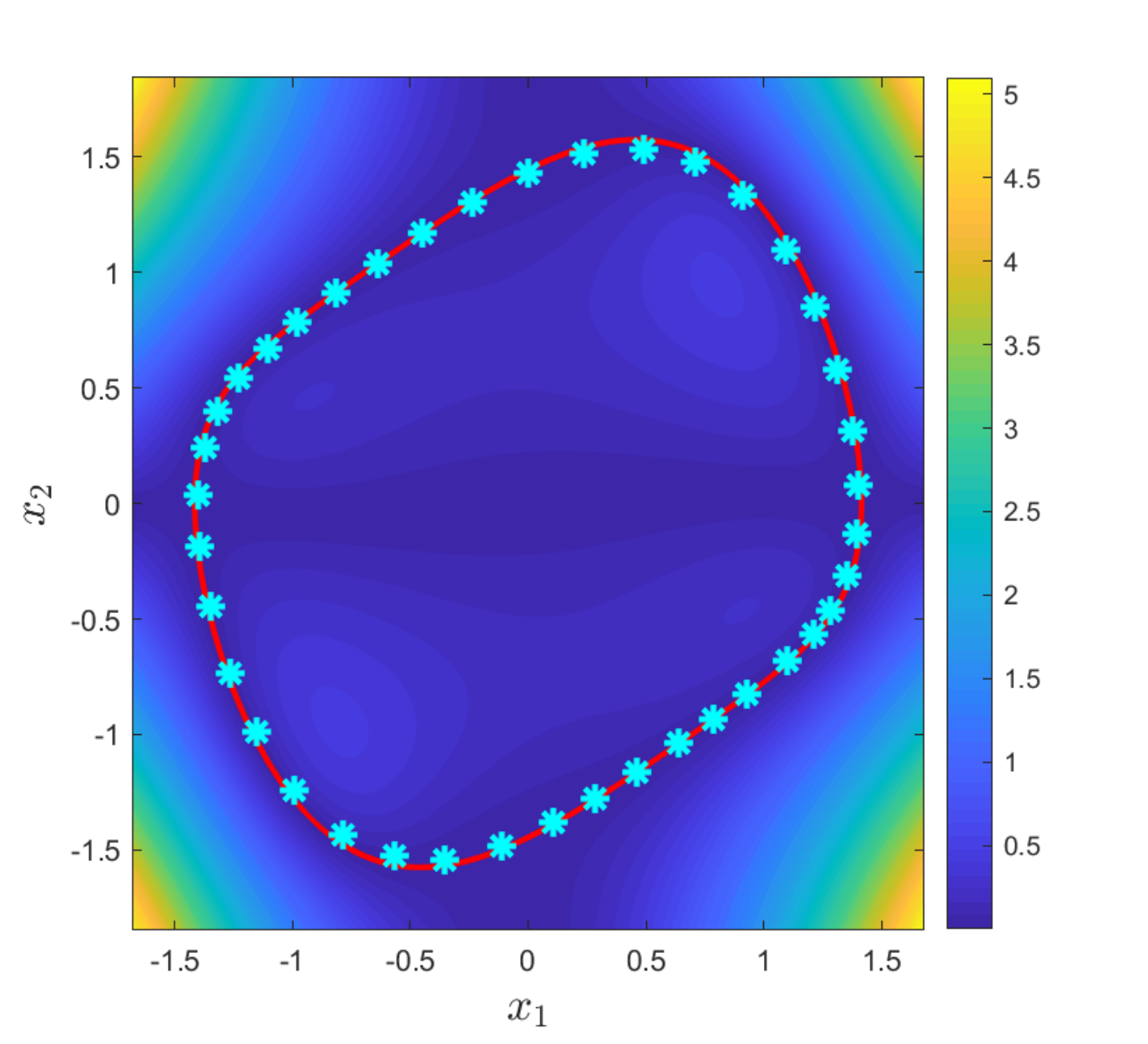}
\caption{Kernel centers for Example \ref{ssec_NumExNonOsc} selected using Algorithm \ref{alg_Koho}  - Pointwise error $|f(\bm{x}) - \hat{f}_n(T, \bm{x})|$ obtained from adaptive estimator. The marker $*$ and the red line represent the kernel centers and the limit set, respectively.}
\label{fig_nonOscpwErrKoho}
\end{figure}
%
%

%

%
\section{Conclusion}
In this paper, we developed criteria for kernel center selection based on the theory of infinite-dimensional adaptive estimation in reproducing kernel Hilbert spaces. We introduced two methods that use this criteria for kernel center selection. These methods provide a simple way to choose kernel centers for a specific class of nonlinear systems - systems in which state trajectory regularly visits the neighborhoods of the positive limit set. We illustrated the effectiveness of both algorithms using practical examples. The approaches discussed in this paper assume a fixed number of kernel centers. It would be of great interest to develop techniques that iteratively add kernel centers in real-time while accounting for the persistence of excitation and fill-distance conditions.

\bibliographystyle{IEEEtran}
\bibliography{j_adap_cent}

\begin{thebibliography}{10}
\providecommand{\url}[1]{#1}
\csname url@samestyle\endcsname
\providecommand{\newblock}{\relax}
\providecommand{\bibinfo}[2]{#2}
\providecommand{\BIBentrySTDinterwordspacing}{\spaceskip=0pt\relax}
\providecommand{\BIBentryALTinterwordstretchfactor}{4}
\providecommand{\BIBentryALTinterwordspacing}{\spaceskip=\fontdimen2\font plus
\BIBentryALTinterwordstretchfactor\fontdimen3\font minus
  \fontdimen4\font\relax}
\providecommand{\BIBforeignlanguage}[2]{{%
\expandafter\ifx\csname l@#1\endcsname\relax
\typeout{** WARNING: IEEEtran.bst: No hyphenation pattern has been}%
\typeout{** loaded for the language `#1'. Using the pattern for}%
\typeout{** the default language instead.}%
\else
\language=\csname l@#1\endcsname
\fi
#2}}
\providecommand{\BIBdecl}{\relax}
\BIBdecl

\bibitem{Ioannou}
P.~A. Ioannou and J.~Sun, \emph{{Robust Adaptive Control}}.\hskip 1em plus
  0.5em minus 0.4em\relax Dover Publications Inc., 1996.

\bibitem{Sastry2011}
S.~Sastry and M.~Bodson, \emph{{Adaptive control: stability, convergence and
  robustness}}.\hskip 1em plus 0.5em minus 0.4em\relax Courier Corporation,
  2011.

\bibitem{Narendra2012}
K.~S. Narendra and A.~M. Annaswamy, \emph{{Stable adaptive systems}}.\hskip 1em
  plus 0.5em minus 0.4em\relax Courier Corporation, 2012.

\bibitem{Kurdila1995}
A.~J. Kurdila, F.~J. Narcowich, and J.~D. Ward, ``{Persistency of excitation in
  identification using radial basis function approximants},'' \emph{SIAM
  journal on control and optimization}, vol.~33, no.~2, pp. 625--642, jul 1995.

\bibitem{Kingravi2012a}
H.~A. Kingravi, G.~Chowdhary, P.~A. Vela, and E.~N. Johnson, ``{Reproducing
  Kernel Hilbert Space Approach for the Online Update of Radial Bases in
  Neuro-Adaptive Control},'' \emph{IEEE Transactions on Neural Networks and
  Learning Systems}, vol.~23, no.~7, pp. 1130--1141, 2012.

\bibitem{Bobade2019}
\BIBentryALTinterwordspacing
P.~Bobade, S.~Majumdar, S.~Pereira, A.~J. Kurdila, and J.~B. Ferris,
  ``{Adaptive estimation for nonlinear systems using reproducing kernel Hilbert
  spaces},'' \emph{Advances in Computational Mathematics}, vol.~45, no.~2, pp.
  869--896, 2019. [Online]. Available:
  \url{https://doi.org/10.1007/s10444-018-9639-z}
\BIBentrySTDinterwordspacing

\bibitem{Mostafa2012}
Y.~S. Abu-Mostafa, M.~Magdon-Ismail, and H.-T. Lin, \emph{{Learning From
  Data}}.\hskip 1em plus 0.5em minus 0.4em\relax AMLBook, 2012.

\bibitem{Wu2012a}
J.~Wu, \emph{{Advances in K-means clustering: a data mining thinking}}.\hskip
  1em plus 0.5em minus 0.4em\relax Springer Science {\&} Business Media, 2012.

\bibitem{Orr1995a}
\BIBentryALTinterwordspacing
M.~J.~L. Orr, ``{Regularization in the Selection of Radial Basis Function
  Centers},'' \emph{Neural Computation}, vol.~7, no.~3, pp. 606--623, may 1995.
  [Online]. Available: \url{https://doi.org/10.1162/neco.1995.7.3.606}
\BIBentrySTDinterwordspacing

\bibitem{Warwick1995a}
K.~Warwick, J.~D. Mason, and E.~L. Sutanto, ``{Centre Selection for Radial
  Basis Function Networks BT - Artificial Neural Nets and Genetic
  Algorithms},'' D.~W. Pearson, N.~C. Steele, and R.~F. Albrecht, Eds.\hskip
  1em plus 0.5em minus 0.4em\relax Vienna: Springer Vienna, 1995, pp. 309--312.

\bibitem{Nie1993}
J.~Nie and D.~A. Linkens, ``{Learning control using fuzzified self-organizing
  radial basis function network},'' \emph{IEEE Transactions on Fuzzy Systems},
  vol.~1, no.~4, pp. 280--287, 1993.

\bibitem{Hager2004}
G.~D. Hager, M.~Dewan, and C.~V. Stewart, ``{Multiple kernel tracking with
  SSD},'' in \emph{Proceedings of the 2004 IEEE Computer Society Conference on
  Computer Vision and Pattern Recognition, 2004. CVPR 2004.}, vol.~1, 2004, pp.
  I--I.

\bibitem{Lin2005}
\BIBentryALTinterwordspacing
G.-F. Lin and L.-H. Chen, ``{Time series forecasting by combining the radial
  basis function network and the self-organizing map},'' \emph{Hydrological
  Processes}, vol.~19, no.~10, pp. 1925--1937, jun 2005. [Online]. Available:
  \url{https://doi.org/10.1002/hyp.5637}
\BIBentrySTDinterwordspacing

\bibitem{Fan2006}
Z.~Fan, M.~Yang, Y.~Wu, G.~Hua, and T.~Yu, ``{Efficient Optimal Kernel
  Placement for Reliable Visual Tracking},'' in \emph{2006 IEEE Computer
  Society Conference on Computer Vision and Pattern Recognition (CVPR'06)},
  vol.~1, 2006, pp. 658--665.

\bibitem{Lian2008a}
J.~Lian, Y.~Lee, S.~D. Sudhoff, and S.~H. Zak, ``{Self-Organizing Radial Basis
  Function Network for Real-Time Approximation of Continuous-Time Dynamical
  Systems},'' \emph{IEEE Transactions on Neural Networks}, vol.~19, no.~3, pp.
  460--474, 2008.

\bibitem{Han2010a}
H.~Han and J.~Qiao, ``{A Self-Organizing Fuzzy Neural Network Based on a
  Growing-and-Pruning Algorithm},'' \emph{IEEE Transactions on Fuzzy Systems},
  vol.~18, no.~6, pp. 1129--1143, 2010.

\bibitem{Han2018}
H.~Han, W.~Lu, Y.~Hou, and J.~Qiao, ``{An Adaptive-PSO-Based Self-Organizing
  RBF Neural Network},'' \emph{IEEE Transactions on Neural Networks and
  Learning Systems}, vol.~29, no.~1, pp. 104--117, 2018.

\bibitem{Han2011}
\BIBentryALTinterwordspacing
H.-G. Han, Q.-l. Chen, and J.-F. Qiao, ``{An efficient self-organizing RBF
  neural network for water quality prediction},'' \emph{Neural Networks},
  vol.~24, no.~7, pp. 717--725, 2011. [Online]. Available:
  \url{http://www.sciencedirect.com/science/article/pii/S0893608011001390}
\BIBentrySTDinterwordspacing

\bibitem{Han2012}
\BIBentryALTinterwordspacing
H.-G. Han, J.-F. Qiao, and Q.-L. Chen, ``{Model predictive control of dissolved
  oxygen concentration based on a self-organizing RBF neural network},''
  \emph{Control Engineering Practice}, vol.~20, no.~4, pp. 465--476, 2012.
  [Online]. Available:
  \url{http://www.sciencedirect.com/science/article/pii/S0967066112000020}
\BIBentrySTDinterwordspacing

\bibitem{Qiao2012}
\BIBentryALTinterwordspacing
J.-F. Qiao and H.-G. Han, ``{Identification and modeling of nonlinear dynamical
  systems using a novel self-organizing RBF-based approach},''
  \emph{Automatica}, vol.~48, no.~8, pp. 1729--1734, 2012. [Online]. Available:
  \url{http://www.sciencedirect.com/science/article/pii/S0005109812002075}
\BIBentrySTDinterwordspacing

\bibitem{Sundararajan2013}
N.~Sundararajan, P.~Saratchandran, and Y.~Li, \emph{{Fully tuned radial basis
  function neural networks for flight control}}.\hskip 1em plus 0.5em minus
  0.4em\relax Springer Science {\&} Business Media, 2013, vol.~12.

\bibitem{Sanner1992}
R.~M. Sanner and J.~.~E. Slotine, ``{Gaussian networks for direct adaptive
  control},'' \emph{IEEE Transactions on Neural Networks}, vol.~3, no.~6, pp.
  837--863, 1992.

\bibitem{Volyanskyy2008}
K.~Y. Volyanskyy, W.~M. Haddad, and A.~J. Calise, ``{A new neuroadaptive
  control architecture for nonlinear uncertain dynamical systems: Beyond
  $\sigma$- and e-modifications},'' in \emph{2008 47th IEEE Conference on
  Decision and Control}, 2008, pp. 80--85.

\bibitem{Kim1999}
Y.~H. Kim and F.~L. Lewis, ``{Neural network output feedback control of robot
  manipulators},'' \emph{IEEE Transactions on Robotics and Automation},
  vol.~15, no.~2, pp. 301--309, 1999.

\bibitem{Patino2002}
H.~D. Patino, R.~Carelli, and B.~R. Kuchen, ``{Neural networks for advanced
  control of robot manipulators},'' \emph{IEEE Transactions on Neural
  Networks}, vol.~13, no.~2, pp. 343--354, 2002.

\bibitem{Nardi2000}
F.~Nardi, ``{Neural network based adaptive alogrithms for nonlinear control},''
  Ph.D. dissertation, 2000.

\bibitem{Senanayake2018a}
\BIBentryALTinterwordspacing
R.~Senanayake, A.~Tompkins, and F.~Ramos, ``{Automorphing Kernels for
  Nonstationarity in Mapping Unstructured Environments},'' in \emph{Proceedings
  of The 2nd Conference on Robot Learning}, ser. Proceedings of Machine
  Learning Research, A.~Billard, A.~Dragan, J.~Peters, and J.~Morimoto, Eds.,
  vol.~87.\hskip 1em plus 0.5em minus 0.4em\relax PMLR, 2018, pp. 443--455.
  [Online]. Available:
  \url{http://proceedings.mlr.press/v87/senanayake18a.html}
\BIBentrySTDinterwordspacing

\bibitem{Chowdhary2010}
G.~Chowdhary and E.~Johnson, ``{Concurrent learning for convergence in adaptive
  control without persistency of excitation},'' in \emph{49th IEEE Conference
  on Decision and Control (CDC)}, 2010, pp. 3674--3679.

\bibitem{Kamalapurkar2017}
R.~Kamalapurkar, B.~Reish, G.~Chowdhary, and W.~E. Dixon, ``{Concurrent
  Learning for Parameter Estimation Using Dynamic State-Derivative
  Estimators},'' \emph{IEEE Transactions on Automatic Control}, vol.~62, no.~7,
  pp. 3594--3601, 2017.

\bibitem{Modares2013}
H.~Modares, F.~L. Lewis, and M.~Naghibi-Sistani, ``{Adaptive Optimal Control of
  Unknown Constrained-Input Systems Using Policy Iteration and Neural
  Networks},'' \emph{IEEE Transactions on Neural Networks and Learning
  Systems}, vol.~24, no.~10, pp. 1513--1525, 2013.

\bibitem{Chowdhary2012}
\BIBentryALTinterwordspacing
G.~Chowdhary, J.~How, and H.~Kingravi, ``{Model Reference Adaptive Control
  using Nonparametric Adaptive Elements},'' in \emph{AIAA Guidance, Navigation,
  and Control Conference}, ser. Guidance, Navigation, and Control and
  Co-located Conferences.\hskip 1em plus 0.5em minus 0.4em\relax American
  Institute of Aeronautics and Astronautics, aug 2012. [Online]. Available:
  \url{https://doi.org/10.2514/6.2012-5038}
\BIBentrySTDinterwordspacing

\bibitem{Chowdhary2013a}
G.~Chowdhary, H.~A. Kingravi, J.~P. How, and P.~A. Vela, ``{Bayesian
  nonparametric adaptive control of time-varying systems using Gaussian
  processes},'' in \emph{2013 American Control Conference}, 2013, pp.
  2655--2661.

\bibitem{Grande2013}
R.~C. Grande, G.~Chowdhary, and J.~P. How, ``{Nonparametric adaptive control
  using Gaussian Processes with online hyperparameter estimation},'' in
  \emph{52nd IEEE Conference on Decision and Control}, 2013, pp. 861--867.

\bibitem{Abdollahi2019}
\BIBentryALTinterwordspacing
A.~Abdollahi and G.~Chowdhary, ``{Adaptive-optimal control under time-varying
  stochastic uncertainty using past learning},'' \emph{International Journal of
  Adaptive Control and Signal Processing}, vol.~33, no.~12, pp. 1803--1824, dec
  2019. [Online]. Available: \url{https://doi.org/10.1002/acs.3061}
\BIBentrySTDinterwordspacing

\bibitem{Liu2018}
M.~Liu, G.~Chowdhary, B.~C. da~Silva, S.~Liu, and J.~P. How, ``{Gaussian
  Processes for Learning and Control: A Tutorial with Examples},'' \emph{IEEE
  Control Systems Magazine}, vol.~38, no.~5, pp. 53--86, 2018.

\bibitem{jia2020a}
J.~Guo, S.~T. Paruchuri, and A.~J. Kurdila, ``{Persistence of Excitation in
  Uniformly Embedded Reproducing KernelHilbert (RKH) Spaces (ACC)},'' in
  \emph{American Control Conference}, 2020.

\bibitem{jia2020b}
\BIBentryALTinterwordspacing
------, ``{Persistence of Excitation in Uniformly Embedded Reproducing Kernel
  Hilbert (RKH) Spaces},'' feb 2019. [Online]. Available:
  \url{https://arxiv.org/abs/2002.07963}
\BIBentrySTDinterwordspacing

\bibitem{Guo2020Rates}
\BIBentryALTinterwordspacing
------, ``{Approximations of the Reproducing Kernel Hilbert Space (RKHS)
  Embedding Method over Manifolds},'' jul 2020. [Online]. Available:
  \url{http://arxiv.org/abs/2007.06163}
\BIBentrySTDinterwordspacing

\bibitem{Aronszajn1950}
\BIBentryALTinterwordspacing
N.~Aronszajn, ``{Theory of Reproducing Kernels},'' \emph{Transactions of the
  American Mathematical Society}, vol.~68, no.~3, pp. 337--404, 1950. [Online].
  Available: \url{http://www.jstor.org/stable/1990404}
\BIBentrySTDinterwordspacing

\bibitem{Berlinet2011}
A.~Berlinet and C.~Thomas-Agnan, \emph{{Reproducing kernel Hilbert spaces in
  probability and statistics}}.\hskip 1em plus 0.5em minus 0.4em\relax Springer
  Science {\&} Business Media, 2011.

\bibitem{Wendland2004}
H.~Wendland, \emph{{Scattered data approximation}}.\hskip 1em plus 0.5em minus
  0.4em\relax Cambridge university press, 2004, vol.~17.

\bibitem{Kurdila2019PE}
\BIBentryALTinterwordspacing
A.~J. Kurdila, J.~Guo, S.~T. Paruchuri, and P.~Bobade, ``{Persistence of
  Excitation in Reproducing Kernel Hilbert Spaces, Positive Limit Sets, and
  Smooth Manifolds},'' sep 2019. [Online]. Available:
  \url{http://arxiv.org/abs/1909.12274}
\BIBentrySTDinterwordspacing

\bibitem{DeVito2012}
\BIBentryALTinterwordspacing
E.~{De Vito}, L.~Rosasco, and A.~Toigo, ``{Learning Sets with Separating
  Kernels},'' apr 2012. [Online]. Available:
  \url{http://arxiv.org/abs/1204.3573}
\BIBentrySTDinterwordspacing

\bibitem{Paruchuri2020Piezo}
\BIBentryALTinterwordspacing
S.~T. Paruchuri, J.~Guo, and A.~J. Kurdila, ``{RKHS Embedding for Estimating
  Nonlinear Piezoelectric Systems},'' feb 2020. [Online]. Available:
  \url{http://arxiv.org/abs/2002.07296}
\BIBentrySTDinterwordspacing

\bibitem{Rasmussen2003}
C.~E. Rasmussen, ``{Gaussian processes in machine learning},'' in \emph{Summer
  School on Machine Learning}.\hskip 1em plus 0.5em minus 0.4em\relax Springer,
  2003, pp. 63--71.

\bibitem{Cortes2004}
J.~Cortes, S.~Martinez, T.~Karatas, and F.~Bullo, ``{Coverage control for
  mobile sensing networks},'' \emph{IEEE Transactions on robotics and
  Automation}, vol.~20, no.~2, pp. 243--255, 2004.

\bibitem{Breitenmoser2010}
A.~Breitenmoser, M.~Schwager, J.~Metzger, R.~Siegwart, and D.~Rus, ``{Voronoi
  coverage of non-convex environments with a group of networked robots},'' in
  \emph{2010 IEEE International Conference on Robotics and Automation}, 2010,
  pp. 4982--4989.

\bibitem{Kohonen2012}
T.~Kohonen, \emph{{Self-organization and associative memory}}.\hskip 1em plus
  0.5em minus 0.4em\relax Springer Science {\&} Business Media, 2012, vol.~8.

\bibitem{Cottrell1998}
\BIBentryALTinterwordspacing
M.~Cottrell, J.~C. Fort, and G.~Pag{\`{e}}s, ``{Theoretical aspects of the SOM
  algorithm},'' \emph{Neurocomputing}, vol.~21, no.~1, pp. 119--138, 1998.
  [Online]. Available:
  \url{http://www.sciencedirect.com/science/article/pii/S0925231298000344}
\BIBentrySTDinterwordspacing

\bibitem{Brocki2007}
{\L}.~Brocki and D.~Kor{\v{z}}inek, ``{Kohonen Self-Organizing Map for the
  Traveling Salesperson Problem},'' R.~Jab{\l}o{\'{n}}ski, M.~Turkowski, and
  R.~Szewczyk, Eds.\hskip 1em plus 0.5em minus 0.4em\relax Berlin, Heidelberg:
  Springer Berlin Heidelberg, 2007, pp. 116--119.

\bibitem{Hubbard2013}
J.~H. Hubbard and B.~H. West, \emph{{Differential equations: A dynamical
  systems approach: Ordinary differential equations}}.\hskip 1em plus 0.5em
  minus 0.4em\relax Springer, 2013, vol.~5.

\end{thebibliography}

%








\end{document}